\documentclass[runningheads]{llncs}
\usepackage[T1]{fontenc}
\usepackage{graphicx}
\usepackage{threeparttable}
\usepackage{booktabs}
\usepackage{tabularx}
\usepackage{enumitem}

\usepackage{multirow}
\usepackage{makecell}
\usepackage{boxedminipage}
\usepackage{amssymb, amsmath}
\usepackage{color}

\usepackage{hyperref}
\begin{document}
\title{List-Decodable Byzantine Robust PIR: Lower Communication Complexity, Higher Byzantine Tolerance, Smaller List Size}
\titlerunning{List-Decodable Byzantine Robust PIR}
%
\author{
Pengzhen Ke\inst{1}
\and
Liang Feng Zhang\inst{1 *}
\and
Huaxiong Wang\inst{2}
\and 
Li-Ping Wang\inst{3}
}
\authorrunning{P. Ke, L. F. Zhang et al.}
%
\institute{School of Information Science and
 Technology, ShanghaiTech University, Shanghai, China 
\\
\email{\{kepzh,zhanglf\}@shanghaitech.edu.cn}
\and
 School of Physical and Mathematical
 Sciences, Nanyang Technological University, Singapore
\\
\email{hxwang@ntu.edu.sg}
\and
Institute of Information Engineering, Chinese Academy of Sciences, Beijing, China
\\
\email{wangliping@iie.ac.cn}
}
\maketitle              
\begin{abstract}
    Private Information Retrieval (PIR) is a privacy-preserving primitive in cryptography. Significant endeavors have been made to address the variant of PIR concerning the malicious servers. 
    Among those endeavors, list-decodable  Byzantine robust  PIR schemes may 
    tolerate a majority of malicious responding servers that  provide incorrect answers. 
        In this paper, we propose two perfect list-decodable BRPIR schemes. Our schemes are the first ones that can simultaneously handle 
      a majority of malicious responding servers, 
      achieve a communication
     complexity of  \( o(n^{1/2}) \) for a database of size $n$, and provide a 
     nontrivial estimation on the list sizes.
     Compared with the existing solutions, our schemes   
     attain  lower communication complexity,
 higher byzantine tolerance, and smaller list size.  

\keywords{Byzantine Robust PIR \and List-decoding Algorithms \and Malicious Servers \and Security.}
\end{abstract}

\section{Introduction} \label{sec:intro}

{
    Private Information Retrieval (PIR) \cite{chor1998private} allows a client to retrieve an element $x_i$ from a database ${\bf x}=(x_1,\ldots,x_n)$ without disclosing to the servers which specific element is being accessed. 
    PIR is a fundamental privacy-preserving primitive in cryptography and has widespread applications in systems such as private media browsing \cite{gupta2016scalable},  metadata-private messaging \cite{angel2016unobservable} and location-based services for smartphones \cite{khoshgozaran2009private,yannuzzi2014key}.
 
    The efficiency of a PIR scheme is mainly measured by its {\em communication complexity}, i.e.,
     the total number of bits that have to be exchanged between 
      the client and all servers  in order to retrieve one bit of the database. 
    A {\em trivial PIR} scheme requires the client to download the entire database from a server. 
    Despite of achieving {\em perfect} privacy, it incurs a prohibitive communication complexity that scales  linearly  in the size $n$ of the database. 
    Chor et al. \cite{chor1998private} showed that the trivial PIR scheme is optimal in terms of communication complexity if there is   only one  server and   perfect privacy is required.
    In order to  achieve  the nontrivial communication complexity of 
    $o(n)$, one has to either   give up the perfect privacy 
    or use multiple servers. One the one hand, by giving up the perfect privacy, single-server PIR 
    schemes \cite{angel2018pir,kushilevitz1997replication,melchor2016xpir,patel2018private} 
    with $o(n)$ communication complexity have
       devised under various computational assumptions. 
On the other hand, 
a {\em multi-server} PIR model \cite{beimel2001information,dvir20162,efremenko20093,woodruff2005geometric,yekhanin2008towards} is necessary if both perfect  privacy and $o(n)$ communication complexity are desired.
    In a $k$-server  PIR scheme, the database $\bf x$ is replicated among $k$ servers and the 
    client retrieves a database element $x_i$ by querying every server once, such that each individual server learns no information about the retrieval index $i$.
    A {\it $t$-private} $k$-server $(t<k)$ PIR scheme \cite{woodruff2005geometric} provides the stronger security guarantee that the retrieval index $i$ is perfectly private for any $t$ colluding servers.

    Compared with the single-server PIR model, the multi-server PIR model fundamentally differs across three critical dimensions:
    \begin{itemize}
        \item \textbf{Privacy}: Single-server PIR schemes rely on computational hardness assumptions, whereas most of the multi-server PIR achieves \emph{information-theoretic} privacy assuming non-collusion.
        
        \item \textbf{Computation Complexity}: Multi-server PIR schemes 
        are free of the heavy public-key operations that are required by single-server PIR and thus much faster (possibly by several orders of magnitude).
        
        \item \textbf{Error Tolerance}: While single-server PIR schemes face complete failure under server compromise, multi-server PIR may still guarantee correct retrieval if some of the servers are controlled maliciously. 
    \end{itemize}
    
\noindent
    These architectural advantages make multi-server PIR an attractive approach, offering stronger security guarantees alongside practical efficiency. However, the involvement of more servers also introduces new challenges. As the number of participating servers increases, so does the risk of receiving incorrect responses—whether due to network failures, outdated data, or deliberate adversarial behavior. Over the past decade, extensive research~\cite{devet2012optimally,goldberg2007improving,kurosawa2019correct,zhang2022byzantine} has focused on combating {\it malicious} servers that may collude and return incorrect answers, potentially leading the client to reconstruct a wrong database entry.
 
    Beimel and Stahl \cite{beimel2002robust} introduced the notion of $b$-{\it Byzantine robust} 
    $k$-out-of-$\ell$ PIR (BRPIR) \cite{banawan2018capacity,tajeddine2019private,zhang2022byzantine},
     which allows a client to retrieve the correct value  if any $k$ out of the 
     $\ell$ servers respond and at most $b$  out of the $k$ responding servers 
     are malicious (called ``Byzantines'') and  return incorrect answers.
     If $b=0$, such schemes are simply called $k$-out-of-$\ell$ robust PIR (RPIR) schemes.  
    They showed that any $a$-out-of-$\ell$ RPIR scheme can be used to construct
     a  $b$-Byzantine robust $k$-out-of-$\ell$ PIR scheme for any $a<k<\ell$ and $b\leq  (k-a)/2$.
    This class of BRPIR schemes can   handle the  cases where a minority of the responding servers
    (i.e., $b < k/2 $) 
    are malicious and  provide incorrect responses.
   However, they require the client to execute  a reconstruction algorithm whose running time may 
   be exponential in $k$, the  number of responding servers. 
  For the same setting of $b<k/2$,
     Kurosawa \cite{kurosawa2019correct} proposed an efficient 
    BRPIR scheme with  polynomial time reconstruction, which 
    is based on the  RPIR scheme of 
     Woodruff and Yekhanin \cite{woodruff2005geometric} and 
     utilizes  the Berlekamp-Welch decoding algorithm \cite{1986Error} 
     for Reed-Solomon codes. 

    While conventional BRPIR schemes are only applicable when $b < k/2$, they are incapable of addressing scenarios where a majority of the responses are erroneous. In practical situations, due to reasons such as (1) outdated databases, (2) poor or unstable mobile networks, and (3) significant communication channel noise, it is likely that most responses could be incorrect. Thus, it becomes imperative to develop a BRPIR scheme that can support higher error tolerance.
    The 
  {\em list-decodable} BRPIR schemes of \cite{devet2012optimally,goldberg2007improving} 
  represent  a class of schemes that incorporate list-decoding techniques from coding theory 
  into  RPIR and enable one to handle the cases of  $b \geq k/2$.
    In list-decodable BRPIR schemes, the client does not reconstruct a single definitive result but a list that contains the correct result. 
    Goldberg \cite{goldberg2007improving} introduced the notion of 
     list-decodable $b$-Byzantine-robust  $t$-private $k$-out-of-$\ell$ PIR and constructed schemes
      that can handle $b$ Byzantine servers for any  \(b < k - \sqrt{kt}\) (Table \ref{table:1}). 
    The $b$-error correction achieved by \cite{goldberg2007improving} is particularly well-suited for Reed-Solomon code-based PIR schemes since it reaches the Johnson bound. 
    The Johnson bound is a crucial threshold that delineates the upper limit of efficient
     polynomial-time list decoding for error-free recovery, ensuring that the maximum
      number of correctable errors in a polynomial-time framework does not exceed 
      this limit. Reaching this bound indicates highly efficient and reliable error 
      correction in managing erroneous responses in the RS code based PIR protocols.
    Devet, Goldberg, and Heninger \cite{devet2012optimally} surpassed the Johnson bound with
     a {\it statistically correct} scheme that allows \( b < k - t - 1 \) \
     (Table \ref{table:1}). 
    However, the higher robustness is achieved at the price of statistical correctness, i.e., 
    allowing a non-zero probability of failure in reconstruction.
    While the state-of-the-art PIR schemes \cite{chee2013query,dvir20162} in the
     honest-but-curious server model achieve a sub-polynomial communication complexity, 
     the list-decodable PIR schemes of \cite{goldberg2007improving,devet2012optimally} exhibit 
     a {\em relatively high communication complexity} of $O(\ell n^{1/2})$. 
Neither  \cite{goldberg2007improving} nor
 \cite{devet2012optimally} gives a good estimation on the 
   maximum size of the list that may be output by the client's 
   reconstruction algorithm. 
    In both works, the  data block $ x_i \in \mathbb
    {F}_p $ of interest 
    is embedded into   a degree-\( t \) polynomial and as the
     constant term. By Johnson Bound,   the number of  polynomials 
     reconstructed by the client's algorithm is at most \( pk^2 \).
      However, the number of the candidate  values for the data block 
      (or equivalently the constant terms of the reconstructed polynomials)
      could be as high as  \( p \). 
    In the worst-case, both schemes might output the entire field 
     \( \mathbb{F}_p \) as the  list of the candidate values for \( x_i \), rendering the schemes ineffective. 

    In most application scenarios of PIR,   the database size $n$ is significantly larger   the parameters $ b,t,k$, and $\ell$.  
    If we restrict  to list-decodable BRPIR schemes that can handle a majority 
    of malicious responding servers (i.e., \( b \geq k/2 \)), the 
    existing solutions have several limitations. 
    First, they cannot  achieve  a communication complexity of \( o(n^{1/2}) \). This limitation 
     persists even in scenarios where a minority  of the responding  servers are Byzantine 
     (i.e., \( b < k/2 \)) or when the privacy threshold $t$ is relatively low (e.g., \( t = O(1) \)). 
Second, they lack a good estimation on the  size of the list that contains the 
block $x_i$ of interest. 
   In this paper, we are interested in  
      list-decodable BRPIR schemes that can simultaneously handle 
      a majority of malicious responding servers (i.e., \( b \geq k/2\)), 
      achieve a communication
     complexity of  \( o(n^{1/2}) \) , and provide a 
     nontrivial estimation on the list sizes.


\begin{table*}[htbp]
    \centering
    \caption{Comparisons for $t$-private list-decodable $k$-out-of-$\ell$ $b$-Byzantine-robust PIR schemes over database in  $\mathbb{F}_{p}^{n}$}
    \label{table:1}
    \begin{threeparttable}
        \begin{tabular}{ccccccc}
            \toprule
            ~~~Reference~~~ & ~~~List-decodable~~~ &  ~~~\makecell{Byzantine\\Bound}~~~ & ~~~\makecell{Communication\\Complexity}~~~  & ~~~\makecell{Maximum\\ List Size}~~ \\
            \midrule
            \cite{goldberg2007improving} & Perfect & $k-\sqrt{kt}$ & $O(\ell n^{1/2})$  & $p$ \\
            \cite{devet2012optimally} & Statistical & $k-t-1$ & $O(\ell n^{1/2} )$  & $p$ 
            \\
            $\Gamma_1$ ({\bf Fig.}  \ref{fig:construction1}) & Perfect & ${ k-2}$  & {$O( {\ell w_1 n^{1/w_1} })$} \tnote{$\dagger$} & $(\frac{k}{k-b})^{w_1 t}$ \\
           $\Gamma_2$ ({\bf Fig.}  \ref{fig:construction2}) & Perfect & $k - \sqrt{kt}$ & $O(\ell w_2  n^{1/w_2})$ \tnote{$\ddagger$}  & $2k$ \\
            \bottomrule
        \end{tabular}
        \begin{tablenotes}
            \item[$\dagger$] \( w_1 \) is a parameter chosen to balance the tradeoff between communication complexity and list size, satisfying \( w_1 < \frac{2k - 2b - 1}{t} \).
            \item[$\ddagger$] \( w_2 \) is a constant defined as \( w_2 = \left\lfloor \frac{(k-b)^2}{kt} \right\rfloor \).
        \end{tablenotes}
    \end{threeparttable}
    \label{table:comparison}
\end{table*}

\subsection{Our Results}   
 In this paper, we construct  two perfect $L$-list decodable $b$-Byzantine-robust 
    $t$-private $k$-out-of-$\ell$  PIR schemes $\Gamma_1$ and $\Gamma_2$ (see {\bf Table} \ref{table:1}). 
In both schemes, the Byzantine robustness  parameter $b$ can be at least $k/2$, the communication complexity
is $o(n)$, and the 
 list size \( L \) solely  depends  on the number $k$ of responding
  servers and the Byzantine robustness \( b \), and is independent of the size $p$ of
  the space where each data block is taken from.

\vspace{2mm}
\noindent
{\bf Byzantine robustness.}
    The scheme $\Gamma_1$ has Byzantine robustness $b < k - \sqrt{kt}$, which is
    comparable to  \cite{goldberg2007improving}. 
 The scheme  $\Gamma_2$ has a Byzantine robustness of $b \leq k-2$
 and 
  surpasses both the Byzantine robustness $b < k-\sqrt{kt}$  of 
 \cite{goldberg2007improving} and the Byzantine robustness $b< k-t-1$ of
    \cite{devet2012optimally}. 

\vspace{2mm}
\noindent
{\bf Communication complexity.}
Compared with  \cite{goldberg2007improving} and 
     \cite{devet2012optimally}, our schemes  achieve a substantially lower communication complexity 
of  \( o(n^{1/2}) \). The communication efficiency of our schemes is 
 particularly high when the privacy threshold       \( t \) is low.
    For instance, with parameters \( (k, b, t) = (20, 12, 1) \), the 
    communication complexity of the schemes \( \Gamma_1 \) and 
  \( \Gamma_2 \) can be  as low as     \( O(n^{1/14}) \) and  \( O(n^{1/4}) \), respectively.
   In contrast, for the same values of $ (k, b, t)$,  the communication complexity of  
  \cite{goldberg2007improving} and \cite{devet2012optimally} is \( O(n^{1/2}) \). 
    More detailed comparisons are provided in {\bf Table.}~\ref{tab:vary k, b, t, F}. 

\vspace{2mm}
\noindent
{\bf List size.}
       The scheme \( \Gamma_1 \) is $L$-list decodable for \(L= \left(  k/(k-b) \right)^{w_1 t} \),
      where \( w_1 < (2k-2b-1)/t  \) is a parameter chosen by the client to balance the
       tradeoff between the communication complexity and the list size. The  scheme 
        \( \Gamma_2 \) is  $L$-list decodable for \( L = 2k \). Furthermore,  when \( b > k/2 \) is large,
the list size of $\Gamma_2$ can be as small as 
         \( L = k \). Since the size of the finite field \( p \) in list-decodable BRPIR is much
          larger than \( k, b, t \),
compared with the $p$-list decodable schemes of 
 \cite{goldberg2007improving} and \cite{devet2012optimally},     
our schemes give   substantial improvements. 
}

\subsection{Background}

\paragraph{\bf From RPIR to List-Decodable BRPIR.}
The starting point of our construction is the robust PIR (RPIR) scheme introduced by Woodruff and Yekhanin~\cite{woodruff2005geometric}. Their work established a powerful framework that reduces the problem of private information retrieval to the task of reconstructing a low-degree univariate polynomial. A detailed description of the Woodruff-Yekhanin RPIR scheme is provided in Section~\ref{sec:Basic Unique decoding PIR of woodruff and Yekhanin}.
In Woodruff-Yekhanin RPIR scheme, the client obtains both the values and the derivatives of a polynomial \( f(\lambda) \) at \( k \) positions, that is, $\{f(\lambda_j), f'(\lambda_j)\}_{j\in [k]}$. The degree of this polynomial is set to be at most \( 2k - 1 \), so that it can be uniquely reconstructed from these evaluations and derivatives. Moreover, the higher the degree of \( f(\lambda) \), the lower the communication complexity of the scheme, which leads to improved efficiency.
A related construction by Beimel and Stahl~\cite{beimel2002robust} addresses the setting with Byzantine servers. When the number of corrupted servers is at most \( b \), and at least \( k - b \) evaluations and derivatives are correct, their scheme reconstructs a polynomial of degree at most \( 2(k - 2b) - 1 \) at each of the \( k - 2b \) positions. A voting-like mechanism is then employed across these reconstructions to recover the correct polynomial \( f(\lambda) \).
However, when \( b \geq k/2 \), unique decoding strategies such as voting are no longer viable. In this case, one must resort to list decoding. 

{
\paragraph{\bf Sudan list decoding algorithm.}
When the number of Byzantine servers exceeds the threshold for unique decoding, list decoding becomes necessary.
Sudan~\cite{sudan1997decoding} proposed one of the earliest list decoding algorithms for Reed--Solomon codes, which can be viewed as a special case of univariate multiplicity codes.
Given \( k \) pairs \( (\lambda_j, \alpha_j) \), the algorithm outputs all low-degree polynomials \( \tilde{f}(\lambda) \) such that \( \tilde{f}(\lambda_j) = \alpha_j \) for at least \( k - b \) indices \( j \in [k] \), where \( b \) is the number of errors.
A detailed description of the Sudan algorithm is provided in Section~\ref{sec:umc}.
}

\subsection{Our Approach}

Building on the polynomial reconstruction framework by Woodruff and Yekhanin in \cite{woodruff2005geometric}, we develop two perfect list-decodable PIR schemes that remain effective even when the fraction of malicious servers exceeds the unique decoding threshold. 
The term ``perfect'' indicates that the list produced by this scheme always contains the queried result.
Our goal is to retain the high-degree structure of the polynomial \( f(\lambda) \), thereby minimizing communication complexity, while extending the scheme’s robustness through list decoding techniques.

Recall that when the number of Byzantine servers \( b \) is less than \( k/2 \), unique decoding remains feasible, and mechanisms like majority voting over polynomial reconstructions at different positions can correctly recover the polynomial $f(\lambda)$. However, once \( b \geq k/2 \), unique decoding breaks down, and we must resort to list decoding—outputting a small list of candidate polynomials that are guaranteed to contain the correct one.

The challenge in this setting is twofold: first, to design a decoding algorithm that produces only a constant-size list (or only polynomial-size in \( k \)), and second, to keep the degree of \( f(\lambda) \) as high as possible, so as to preserve the communication efficiency inherited from the Woodruff-Yekhanin RPIR structure. To address the challenge, we introduce two distinct decoding strategies in our constructions, each leading to a perfect list-decodable PIR scheme.

In the first scheme, \( \Gamma_1 \), we introduce a method called \emph{overinterpolation}. This method imposes a strict upper bound on the degree of any admissible interpolated polynomial, allowing us to discard inconsistent candidates while ensuring that the correct polynomial remains in the list.  
For example, we interpolate a polynomial from each subset of \( k - b - 1 \) out of \( k \) values and retain only those polynomials whose degree is at most \( (k - b) / 2 \).  
By carefully tuning this degree threshold, we can effectively bound the list size within a polynomial in \( k \).
This method is conceptually simple and yields both a small list size and low computational complexity when \( \binom{k}{b} \) is small. However, when \( \binom{k}{b} \) becomes large—if \(b\) and \(k - b\) are both in $O(k)$ —the computational cost grows exponentially with \( k \).  
To address such a situation, we further develop an optimized algorithm that reduces the overall complexity to a polynomial in \( k \), even in such parameter settings.

In the second scheme, \( \Gamma_2 \), 
we propose a new list‐decoding algorithm for order-$1$ multiplicity codes that follows Sudan list decoding algorithm.
Unlike Sudan list decoding algorithm, which operates on Reed--Solomon codewords consisting of the point-value pair \( (\lambda, f(\lambda)) \), our algorithm works with multiplicity codewords composed of the point-value-derivative tuple \( (\lambda, f(\lambda), f'(\lambda)) \).
Given \( k \) tuples \( (\lambda_j, \alpha_j, \beta_j) \), our list decoding algorithm constructs a pair of polynomials: a base interpolating polynomial \( Q^{\text{base}}(\lambda, \alpha) \), and an extended polynomial \( Q^{\text{ext}}(\lambda, \alpha, \beta) \) that encodes additional derivative information. Following a process analogous to Sudan's approach, we identify a list of candidate polynomials \( \tilde{f}(\lambda) \) that are consistent with at least \( k - b \) of the given evaluations and derivatives.
A key requirement of this approach is that the degree of the target polynomial \( f(\lambda) \) must not exceed \( (k - b)^2 / k \). Compared to the Woodruff-Yekhanin RPIR scheme, where the degree reaches \( 2k - 1 \), this constraint results in a higher communication overhead. Nonetheless, it enables decoding in the environments where unique decoding is no longer feasible, while ensuring that both the list size and decoding time remain polynomial in \( k \).

Both constructions can be seen as perfect list-decodable generalizations of the Woodruff-Yekhanin scheme, and are particularly suited for the environments in which unique decoding is provably impossible.


\subsection{Related Work}
    The investigation of Byzantine robust PIR (BRPIR) schemes and list-decodable BRPIR schemes has made significant strides in addressing the challenges posed by malicious servers in PIR scenarios.

  \vspace{2mm}
  \noindent
    {\bf Byzantine Robust PIR (BRPIR)} schemes\cite{banawan2018capacity,beimel2002robust,kurosawa2019correct,tajeddine2019private,zhang2022byzantine} enable the client to both correctly retrieve the desired database element despite the presence of a limited number of malicious servers but also allow the client to identify which servers are acting maliciously. 
    Beimel first introduced the idea of Byzantine robustness in \cite{beimel2002robust}.     
    Given an $a$-out-of-$\ell$ robust PIR scheme, they provided a general construction from an $a$-out-of-$\ell$ robust PIR scheme to a $k$-out-of-$\ell$, $b$-Byzantine-robust PIR scheme for any $a<k<\ell$ with $b = (k-a)/2$. 
    Applying this construction to the robust PIR scheme proposed by Woodruff and Yekhanin \cite{woodruff2005geometric}, which retrieves hidden data blocks by interpolating polynomials, one can obtain an efficient $t$-private $k$-out-of-$\ell$ $b$-Byzantine-robust PIR scheme, where $b < k/2$, and the communication complexity of this scheme is $O(n^{1/ \lfloor (2(k-2b)-1)/t \rfloor})$. However, this scheme requires the client to perform computations that grow exponentially with the number of responding servers $k$, leading to high computational complexity for the client.
    By extending the Berlekamp-Welch algorithm \cite{1986Error}, a decoding algorithm for Reed-Solomon (RS) codes, to the case of first-order derivatives, Kurosawa \cite{kurosawa2019correct} reduced the client-side computation of the construction by Beimel and Stahl to polynomial levels. Nonetheless, the Byzantine tolerance remains bounded by $b < k/2$.

      \vspace{2mm}
  \noindent
    {\bf List-Decodable BRPIR} schemes \cite{devet2012optimally,goldberg2007improving} are suitable for scenarios where conventional BRPIR schemes fail to operate, particularly when $b \geq k/2$.  
    This idea of list-decoding in BRPIR was first introduced by Goldberg \cite{goldberg2007improving} to improve the bound of byzantine tolerance $b$. 
    Goldberg proposed a $t$-private list-decodable $k$-out-of-$\ell$ $b$-Byzantine-robust PIR scheme for any $b < k - \sqrt{kt}$. 
    In Goldberg's scheme, the database is structured as a matrix. 
    The client retrieves a row of this matrix from each server in the form of a degree-$t$ RS codeword. 
    The data block is then reconstructed using the Guruswami-Sudan list decoding algorithm \cite{guruswami1998improved}, which is a list decoding algorithm for RS codes, applied to the codewords of RS codes.   
    For cases where $t < k/4$, this scheme can achieve $b \geq k/2$. 
    Devet, Goldberg and Heninger \cite{devet2012optimally} improved the Byzantine bound of list-decodable BRPIR by proposing a {\em statistical} list-decodable BRPIR scheme with byzantine bound $b < k - t - 1$.
    In both works of \cite{goldberg2007improving} and \cite{devet2012optimally}, the client’s desired data block \( x_i \in \mathbb{F}_p \) is embedded in the constant term of a degree-\( t \) polynomial, which is then subjected to list decoding. According to the bound given by the Johnson Bound, the number of decoded polynomials is at most \( pk^2 \). However, the possible values for the data block are at most \( p \). 
    In the worst-case scenario, both schemes might output the entire \( \mathbb{F}_p \) as the candidate list for \( x_i \), rendering the schemes ineffective.


    \vspace{1mm}

    There are other works addressing malicious servers that focus solely on error detection without aiming to retrieve the correct result. Some of these approaches achieve higher efficiency and greater Byzantine server tolerance compared to BRPIR. 
    
  \vspace{2mm}
  \noindent
    {\bf Verifiable PIR (VPIR)} schemes in both the multi-server model \cite{zhang2014verifiable} and the single-server model \cite{zhao2021verifiable}, ensures that the client can identify the byzantine servers, that is, can tell which server is malicious. However, it does not guarantee the recovery of the correct element, or even a small list of potential candidates. The security of VPIR is weaker than that of BRPIR. However, this allows VPIR to accommodate a significantly higher number of malicious servers and greatly reduces communication complexity.

       \vspace{2mm}
  \noindent
    {\bf Error Detecting PIR (EDPIR)} \cite{cao2023committed,colombo2023authenticated,eriguchi2022multi,ke2022two,ke2023private,kruglik2023querying,zhu2022post} schemes allow the client to detect the existence of incorrect answers provided by malicious servers, though they do not guarantee identifying which servers are malicious or reconstructing the correct value.
    \vspace{1mm}    

    In our work, we devise a list decoding method for a variant of Reed-Solomon (RS) codes that consider first-order derivatives. By integrating this list decoding method into the Woodruff-Yekhanin robust PIR scheme, we obtain a list-decodable BRPIR scheme. This variant of RS codes, which incorporates first-order derivatives, can be viewed as a special case of univariate multiplicity codes.

  \vspace{2mm}
  \noindent
    {\bf Univariate Multiplicity Codes}~\cite{guruswami2008explicit} are variants of Reed–Solomon (RS) codes that are constructed by evaluating a polynomial along with all of its derivatives up to order $s$. Notably, RS codes correspond to the special case when $s = 0$. In recent years, there has been significant progress in the list decoding of univariate multiplicity codes~\cite{goyal2024fast,guruswami2011optimal,guruswami2013linear,kopparty2015list,kopparty2023improved}, with the decoding radius---corresponding to the Byzantine tolerance $b/k$ in our framework---reaching up to $1 - R - \epsilon$ for any $\epsilon > 0$, where $R$ is the code rate. This matches the list decoding capacity $1 - R$. 
    However, when restricted to only first-order derivatives and function values, i.e., order-$1$ multiplicity codes, these results no longer apply. We provide a detailed discussion of these limitations in Section~\ref{sec:umc}.


\section{Preliminaries} \label{sec:preliminaries}
{
\subsection{Notation} \label{sec:notation}

We use bold lower-case letters to denote vectors. 
    For any vector \({\bf v}\)  (resp. \({\bf v}_j\)) of length \(m\) and any  $c\in[m]$, we denote by \(v_c\) (resp. \(v_{j,c}\))   the \(c\)-th element of \({\bf v}\) (resp.  \({\bf v}_j\)), meaning that
     \({\bf v}=(v_1,\ldots,v_m)\)  (resp. \({\bf v}_j=(v_{j,1},\ldots,v_{j,m})\)). 
        For any finite set $A$, we denote by $|A|$ the {\em cardinality} of  $A$.
    For any integer $n>0$, we denote   $[n] = \{1,2,\ldots, n\}$ and   $\{a_j\}_{j\in [n]}=\{a_1, a_2, \ldots, a_n\}$.
    For any two   vectors  ${\bf u},{\bf v}$ of the same length, we denote by $\left<{\bf u} , {\bf v} \right>$ the {\em inner product} of ${\bf u}$ and ${\bf v}$.
    For any prime $p$, we denote by $\mathbb{F}_p$ the {\em finite field} of $p$ elements and 
denote by $\mathbb{F}_p^n$ the   set of  all length-$n$ vectors over 
$\mathbb{F}_p$. 
     For any polynomial $f(\lambda)\in \mathbb{F}_p[\lambda]$ and any integer $s\geq 0$, we denote by $f^{(s)}(\lambda)$ the {\em order-$s$ derivative} of $f$ with respect to $\lambda$ and      denote by 
    $f^{(\leq s)}(\lambda)=(f^{(0)}(\lambda), f^{(1)}(\lambda),...,f^{(s)}(\lambda))$ the 
    {\em order-$s$ evaluation} of $f$ at $\lambda$.  In particular,    $f^{(0)}(\lambda)=f(\lambda)$ and 
    $ f^{(1)}(\lambda)=f'(\lambda)$ are the evaluation and  order-$1$ derivative of $f$ at $\lambda$, respectively.

    \vspace{1.5mm}

    We will use the following variables throughout the paper: 
\begin{itemize}
    \item $\ell$: the total number of servers.
    \item $k$: the number of responding servers.
    \item $t$: the number of servers that may collude to learn the retrieval index. 
    \item $b$: the number of servers that may collude to respond incorrectly. 
    \item $n$: the database size.
    \item $p$: the database block size, each entry in the database is an element in $\mathbb{F}_p$.
    \item ${\bf x} = (x_1,\ldots,x_n)$: the database, which is a vector  in $\mathbb{F}_p^n$.
    \item $w$: a degree parameter.
    \item $m$: the least positive integer such that ${m\choose w}\geq n$. 
    \item $C$: a subset of  $\mathbb{F}_p^m$ and a   constant weight code of length $m$ and weight $w$. 
    \item $E$:  a {\em public} $1$-to-$1$ encoding function. 
    \item $F$: an encoding of the database ${\bf x}$ such that $F(E(i)) = x_i$.
\end{itemize}

\subsection{Constant-weight Code} \label{sec:constant-weight code}

A {\em code} $C$ of length $n$ over $\mathbb{F}_p$ is a subset of $\mathbb{F}_p^n$.  
For any two {\em codewords} ${\bf u,v}\in C$,    the {\em Hamming distance} between 
${\bf u,v}$ is the number of coordinates where ${\bf u,v}$ differ and denoted by
$d_H({\bf u,v})=|\{i\in [n]: {\bf u}_i\neq {\bf v}_i\}|$. 
The {\em minimum distance} of  a code $C$ is 
 the least Hamming distance between any two different
codewords in $C$ and denoted by  
 $d_H(C)=\min_{{\bf u,v}\in C, {\bf u}\neq {\bf v}}
d_H({\bf u,v})$. 
For any codeword ${\bf u}\in C$,   the Hamming  weight of $\bf u$ is
the  number nonzero coordinates of
${\bf u}$ and denoted by ${\sf wt}({\bf u})=d_H({\bf u}, {\bf 0})$.

    For any integers $m,d,w>0$, an $(m,d,w)$ {\em constant-weight code} \cite{bose1982theory} is a binary code of length $m$, minimum Hamming distance $d$, and Hamming weight $w$.
    The {\em maximum} size of an $(m,d,w)$ constant weight code is denoted by
    $A(m,d,w)$. 
    For any integers $m,\delta,w>0$, Agrell, Vardy and Zeger\cite{agrell2000upper} showed 
    an   upper bound on  $A(m,2\delta,w)$.
    \begin{theorem} \label{theo:constant-weight-code}
        {\em (Agrell, Vardy and Zeger \cite{agrell2000upper}, Theorem 12)}
        $$ 
            A(m,2\delta,w)\leq \frac{{m \choose w - \delta + 1}}{{w \choose w - \delta + 1}}.
        $$
    \end{theorem}

\subsection{Univariate Multiplicity Codes and their List-decoding Algorithms} \label{sec:umc}
    For any integer $s\geq 0$, an \emph{order-\(s\) univariate multiplicity code} \(C\) of length \(k\) for degree-\(w\) polynomials over \(\mathbb{F}_p\) is an error-correcting code that encodes each polynomial \(f(\lambda)\in\mathbb{F}_p[\lambda]\) with \(\deg f\le w\) into the codeword
    \[
      C_f=\bigl(f^{(\le s)}(\lambda_j)\bigr)_{j=1}^k,
    \]
    where \(\lambda_1,\dots,\lambda_k\) are pairwise distinct elements of \(\mathbb{F}_p\).  In particular, the renowned Reed–Solomon codes arise when \(s=0\).  In this work, we specialize to \(s=1\).

    Given a set of $k$ tuples $\{(\lambda_j, \alpha_{0,j}, \ldots, \alpha_{s,j})\}_{j=1}^k$, a 
    {\em list decoding} algorithm that corrects $b$ errors for a univariate multiplicity code $C$ may  identify all polynomials $\tilde{f}(\lambda)$ of degree  $\leq w$ such that $\tilde{f}^{(\leq s)}(\lambda_j) = (\alpha_{0,j}, \ldots, \alpha_{s,j})$ for at least $k-b$ distinct indices $j\in [k]$ and  outputs the list of all such polynomials. 
    If the size of the list is at most $L$, then the code is called {\em $L$-list decodable}.

    For Reed-Solomon codes (i.e., $s=0$), Sudan \cite{sudan1997decoding} has a list decoding algorithm that interpolates a bivariate polynomial 
    \begin{align}\label{eqn:sudan}
            Q(\lambda, \alpha) = \sum_{c= 0}^{D/w} Q_{c}(\lambda) \alpha^c
    \end{align}
  of $(1, w)$-weighted degree  $\leq D$
  by solving  $k$   constraints  of the form 
  $$Q(\lambda_j, \alpha_{0,j}) = 0, j\in [k],$$
  and then outputs a list of candidate polynomials
  $\tilde{f}(\lambda)$ such that 
  \( (\alpha - \tilde{f}(\lambda)) |Q(\lambda,\alpha),\)
  where  $D$ is a carefully chosen parameter and 
 $\deg(Q_c(\lambda))\leq D-cw$ for all $0\leq c \leq D/w$.

Recent  list decoding algorithms
\cite{goyal2024fast,guruswami1998improved,guruswami2011optimal,guruswami2013linear,kopparty2015list,kopparty2023improved} for the general order-$s$ univariate multiplicity codes 
extended the basic idea \cite{sudan1997decoding} of 
{\em interpolating} a multivariate 
polynomial $Q$ and then {\em factoring} $Q$ to find  all candidate polynomials
$\tilde{f}(\lambda)$. 
An early list-decoding algorithm of univariate multiplicity codes given by Guruswami and Wang \cite{guruswami2011optimal,guruswami2013linear} considered 
\begin{align} \label{eqn:gw}
        Q(\lambda, \alpha_0, \ldots, \alpha_r) = P(\lambda) + \sum_{c=0}^{r} Q_c(\lambda)\alpha_c,
\end{align}
  an  $(r+2)$-variate polynomial   of degree $\leq D$, where $r \leq s$ and $D$ are carefully chosen parameters, 
 \( P(\lambda) \) is of   degree   \( \leq D \), and  \( Q_c(\lambda) \) is of
   degree   \(\leq  D-w+1 \) for all \( c \in \{0, \ldots, r\} \). 
In particular, the polynomial $Q$ may be viewed as a function of 
$\lambda$ if $\alpha_0,\ldots,\alpha_s$ are all viewed as  functions of $\lambda$. They constructed $s-r$ new polynomials $\{\mathcal{D}^{c}Q(\cdot)\}_{c=1,\ldots, s-r}$ from $Q$ by successively applying a special  operator
  $\mathcal{D}$ that essentially differentiates   $Q$ with respect to $\lambda$ using the well-known
   chain rule but requires that
  $\mathcal{D}(\alpha_c)=\alpha_{c+1}$ (instead of $\mathcal{D}(\alpha_c)=\alpha_c'$) for all $c=0,1,\ldots,s-1$. 
They developed    $(s-r + 1)k$ constraints that
\[
    Q(\lambda_j, \alpha_{0,j}, \ldots,\alpha_{r,j}) = 0 ~~\text{and}~~\mathcal{D}^{c}Q(\lambda_j, \alpha_{0,j}, \ldots, \alpha_{s,j}) = 0,~c \in [s-r], j\in[k].
\]
In their construction, the parameter \(D\) is prescribed as
\[
  D = \left\lfloor \frac{k(s - r + 1) - w}{r + 1} \right\rfloor,
\]
which provides resilience against up to
$
  b = k - \left\lfloor \frac{D + w}{s - r + 1} \right\rfloor
$
errors.
The advantage of this construction lies in its ability to achieve a trade‐off between computational complexity and error tolerance by appropriately choosing the parameters $r$ and $D$. However, in the special case $s=1$, which is crucial for our construction of list‐decodable BRPIR schemes, it can correct at most $b<k/2$ errors, regardless of the choice of $r = 0$ or $ r = 1$.
}

The recent works in~\cite{goyal2024fast,kopparty2023improved} establish that order-$s$ univariate multiplicity codes of length $k$, encoding degree-$w$ polynomials over $\mathbb{F}_p$, are $L$-list decodable from up to $b = \left(1 - \frac{w}{sk} - \epsilon\right)k$ errors, where $L = \left(\frac{1}{\epsilon}\right)^{O\left(\frac{1}{\epsilon} \log \frac{1}{\epsilon}\right)}$ and $\epsilon \geq \sqrt{16/s}$.
When $s$ is large, the list decoding error tolerance $b/k$ of this scheme approaches the capacity $1 - R$, where $R = w/sk$ is the code rate. However, for small values of $s$, such as $s = 1$, the constraint $\epsilon > 1$ renders this result inapplicable.

{
\subsection{Woodruff-Yekhanin RPIR Scheme} \label{sec:Basic Unique decoding PIR of woodruff and Yekhanin}

    \begin{figure*}[htbp!]
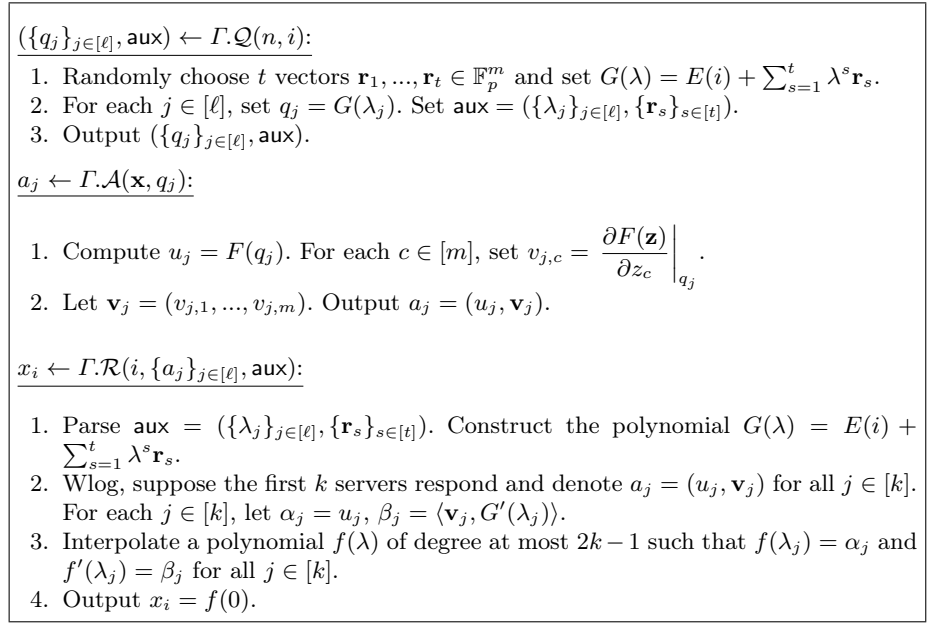

    \begin{center}
		\begin{boxedminipage}{\textwidth}
\vspace{2mm}
\underline{$(\{q_j\}_{j\in [\ell]}, {\sf aux}) \leftarrow \Gamma.\mathcal{Q}(n, i)$:}
\vspace{0.5mm}
\begin{itemize}
    \item[1.] Randomly choose $t$ vectors ${\bf r}_{1},...,{\bf r}_{t} \in \mathbb{F}_p^m$ and set $G(\lambda) = E(i) + \sum_{s = 1}^{t} \lambda^s {\bf r}_s$.

    \item[2.] For each $j \in [\ell]$, set $q_{j} = G(\lambda_{j})$. Set ${\sf aux} = ( \{\lambda_{j}\}_{j\in [\ell]}, \{{\bf r}_s\}_{s\in [t]} )$.

    \item[3.] Output $(\{q_j\}_{j\in [\ell]}, {\sf aux})$.
\end{itemize}

\vspace{2mm}
\underline{$a_j \leftarrow \Gamma.\mathcal{A}({\bf x}, q_j)$:}
\vspace{0.5mm}
\begin{itemize}
    \item[1.] Compute $u_j = F(q_j)$. For each $c \in [m]$, set
    $\displaystyle v_{j, c} = \left. \frac{ \partial F({\bf z}) }{\partial z_{c} }\right|_{q_j}$.

    \item[2.] Let ${\bf v}_j = (v_{j,1}, ..., v_{j,m})$. Output $a_j = (u_j, {\bf v}_j)$.
\end{itemize}

\vspace{2mm}
\underline{$x_i \leftarrow \Gamma.\mathcal{R}(i, \{a_j\}_{j\in [\ell]}, {\sf aux})$:}
\vspace{0.5mm}
\begin{itemize}
    \item[1.] Parse ${\sf aux} = (\{\lambda_j\}_{j\in [\ell]}, \{{\bf r}_s\}_{s\in [t]})$. Construct the polynomial $G(\lambda) = E(i) + \sum_{s = 1}^{t} \lambda^s {\bf r}_s$.

    \item[2.] Wlog, suppose the first $k$ servers respond and denote $a_j = (u_j, {\bf v}_j)$ for all $j \in [k]$.  
    For each $j \in [k]$, let $\alpha_j = u_j$, $\beta_j = \left< {\bf v}_{j}, G'(\lambda_j) \right>$.

    \item[3.] Interpolate a polynomial $f(\lambda)$ of degree at most $2k-1$ such that $f(\lambda_j)= \alpha_j$ and $f'(\lambda_j) = \beta_j$ for all $j \in [k]$.

    \item[4.] Output $x_i = f(0)$.
\end{itemize}

            \end{boxedminipage}
	\end{center}
	\caption{Woodruff-Yekhanin $(t,k,\ell)$-RPIR scheme $\Gamma$.}
	\label{fig:WY PIR}
\end{figure*}

    Woodruff and Yekhanin    \cite{woodruff2005geometric}  constructed a $t$-private $k$-out-of-$\ell$ RPIR scheme, referred to as a $(t,k,\ell)$-RPIR scheme, $\Gamma$ (see {\bf Fig.} \ref{fig:WY PIR}) with communication complexity 
    $
        \mathcal{O}(\frac{k \ell}{t} \log \ell \cdot n^{1/ \lfloor(2k-1)/t \rfloor }),
    $
    where $n$ is the size of the database ${\bf x}=(x_1,\ldots,x_n)$. 

 For a degree parameter $w$ defined as     
 \begin{equation} \label{eq: wy w}
     w =  \lfloor (2k-1)/t \rfloor,
 \end{equation}
 their scheme   chooses  an integer $ m = O(w n^{1/ w})$ such that
  $        {m \choose w} \geq n,$
  encodes the indices of the $n$ database elements as binary vectors of length $m$ and weight 
  $w$ with a {\em public} $1$-to-$1$ function
    \[
        E:[n] \to \{0,1\}^m,
    \]
    and represents the database  ${\bf x}$ as 
    \begin{equation} \label{eqn:Fx}
        F(\textbf{z}) = F(z_1,\ldots, z_m) = \sum_{j=1}^{n} x_j \cdot \prod_{c:E(j)_c = 1} z_{c},
    \end{equation}
    a homogeneous $m$-variate polynomial of degree $w$ that satisfies 
    \begin{equation}
        F(E(i)) = x_i
    \end{equation}
    for all $i\in [n]$. The client reduces the problem of privately retrieving $x_i$ from $\ell$ servers to the problem of privately evaluating $F(E(i))$ with the $\ell$ servers.
    To this end, a prime $p > \ell$ is chosen and $F$ is interpreted as a polynomial over the finite field $\mathbb{F}_p$.
    For each $j\in [\ell]$, the $j$-th server is associated with a nonzero field elements $\lambda_j$. In particular, the field elements $\{\lambda_j\}_{j\in [k]}$ are distinct and can be made public.
To retrieve $x_i$, the client chooses $t$ vectors ${\bf r}_1,{\bf r}_2,\ldots,{\bf r}_t \in \mathbb{F}_p^m$ uniformly at random and generates a vector  
    \begin{equation}
        \begin{split}
            G(\lambda) = E(i) + \sum_{s = 1}^{t} \lambda^{s} {\bf r}_{s}
        \end{split}
    \end{equation}
  of   $m$ polynomials of degree $t$. 
It   sends a query 
$$q_j = G(\lambda_j)$$ to the $j$-th server for all 
  $j\in [\ell]$ and keeps the information  ${\sf aux} = ( \{\lambda_j\}_{j\in [\ell]}, \{{\bf r}_s\}_{s \in [t]} )$ for later use. The $j$-th server  is expected to 
  compute    $u_j = F(q_j)$ and ${\bf v}_j=(v_{j,1},\ldots,v_{j,m})$, where 
    \begin{equation} \nonumber
        \begin{split}
            &v_{j, c}=  \left.\frac{\partial F({\bf z})}{\partial z_{c}} \right|_{q_j}
        \end{split}
    \end{equation}
    for all $c \in [m]$, and reply with 
    $$a_j = (u_j, {\bf v}_j ).$$ 
    Wlog, suppose the first $k$ servers respond, the client  interpolates  the univariate polynomial 
    \begin{equation} \label{eqv1}   
        f(\lambda)  = F(G(\lambda))
    \end{equation}
    of degree $ wt\leq 2k-1$ from the $2k$ values 
    \[
        f(\lambda_j) = u_j, ~~
        f'(\lambda_j) =  \left< {\bf v}_j, G'(\lambda_j) \right>, ~~j\in [k]    \]
    and outputs $f(0)$, which is equal to $x_i$ as
    \[
        x_i =F(E(i))= F(G(0))=f(0). 
    \]

 The scheme is $t$-private such that any $t$ colluding servers learn no information about
 $i$ (or equivalently $E(i)$), because every element of 
 $E(i)$ is secret-shared among the servers with
 Shamir's $t$-private threshold secret sharing scheme \cite{shamir1979share}, using the vector
 $G(\lambda)$ of $m$ random polynomials. 
    
        
Note that the client sends a length-$m$ vector in $\mathbb{F}_p$ to each server and each server returns a length-$(m + 1)$ vector in $\mathbb{F}_p$.
Given that $m = O(w n^{1/w})$, if   the  prime $p$ is chosen    such that $\ell < p \leq 2\ell$, then the 
 communication complexity of this scheme is      $(2m+1)\ell \log p=O(\frac{k \ell \log \ell}{t} n^{1/\lfloor \frac{2k-1}{t}\rfloor} )$.
    
In the default configuration of  \cite{woodruff2005geometric},     $\lambda_j$ is set to $j$
 for all $j\in[\ell]$ and eliminates the necessity   to  include $\{\lambda_j\}_{j=1}^\ell$ in ${\sf aux}$.

\section{List-Decodable BRPIR model}
    In this section, we formally define a model for  list-decodable BRPIR, which 
     generalizes the standard  BRPIR  of \cite{beimel2002robust} by allowing 
     the reconstructing algorithm to output a list that contains the data item being  
     retrieved.

Similar to the standard BRPIR,  the list-decodable BRPIR allows   a client to 
retrieve a data item from a database replicated among multiple servers, even if some of
the servers are silent or malicious, where the silent servers simply 
  fail to respond and the malicious servers  respond incorrectly.
  While the standard BRPIR may enable 
  correct retrieval when a minority of the responding servers are malicious, what really differentiates 
  our list decodable BRPIR from the standard BRPIR
  is its ability to handle the much trickier case that a majority of the responding servers are
  malicious and thus  offer  stronger  robustness for more adversarial environments.
 
Informally, an $L$-list-decodable  $b$-Byzantine-robust $t$-private $k$-out-of-$\ell$ PIR scheme  is a  protocol between  $\ell$ servers $\{\mathcal{S}_j\}_{j\in [\ell]}$, each storing  a copy of the same database 
 ${\bf x} = (x_1,\ldots,x_n)\in \mathbb{F}_p^n$, and a client $\mathcal{C}$ that is interested in 
    a block  $x_{i}$ of the database. 
    It allows the client to output a list of size $\leq L$ that contains 
    $x_i$, as long as at least $k$ out of the   $\ell$ servers   respond and
     at most $b$ out of the  responding servers provide incorrect answers; and 
     
     \begin{definition}[{\bf List-decodable BRPIR}]
     \it
        An {\em $L$-list-decodable   $b$-Byzantine-robust $t$-private $k$-out-of-$\ell$ PIR} scheme $\Gamma = (\mathcal{Q}, \mathcal{A}, \mathcal{R})$ for a client $\mathcal{C}$ and 
          $\ell$ servers $\mathcal{S}_1,\ldots,\mathcal{S}_\ell$ consists of three algorithms that can be described as follows:
        \begin{itemize}
            \item $(\{q_j\}_{j \in [\ell]}, {\sf aux})\leftarrow \mathcal{Q}(n,i)$:
                This is a randomized {\em querying algorithm} for the client $\mathcal{C}$. It takes the database size $n$ and a retrieval index $i\in [n]$ as input, and outputs $\ell$ queries $\{q_j\}_{j \in [\ell]}$, along with an auxiliary information ${\sf aux}$.
                For each $j\in [\ell]$, the query $q_j$ will be sent to the server $\mathcal{S}_j$. The auxiliary information ${\sf aux}$ will be used later by the client
           for reconstruction. 

            \item $a_j \leftarrow \mathcal{A}({\bf x},q_j)$:
                This is a deterministic {\em answering algorithm} for the server $\mathcal{S}_j$ $(j\in[\ell])$. It takes a database ${\bf x} = (x_1,\ldots,x_n)$ and the query $q_j$ as input and outputs an answer $a_j$.

            \item ${\sf output\_list} \leftarrow \mathcal{R}(i, \{a_j\}_{j \in [\ell]}, {\sf aux})$:
                This is a deterministic {\em reconstructing algorithm} for the client $\mathcal{C}$. It uses the retrieval index $i$, the answers $\{a_j\}_{j \in [\ell]}$ and the auxiliary information ${\sf aux}$ to reconstruct $x_i$, where $a_j$ is set to 
                    $\perp$ if a server $\mathcal{S}_j$ does not respond. The output is a list of
                    size at most $L (= k^{O(1)})$   that 
                    contains the queried data block $x_i$. 
        \end{itemize}

 \noindent
 {\bf Correctness.}  
 Informally, 
the scheme $\Gamma$ is considered  {\em correct} if
the ${\sf output\_list}$ generated by the reconstructing algorithm $\mathcal{R}$ is always  of size 
  $\leq L$ and includes the target block $x_i$, provided that at least $k$ out of the $\ell$ servers respond 
  and   at most $b$ of the responses are incorrect. 
Formally, the scheme $\Gamma$ is   {\em correct} if for any $n$, any ${\bf x} \in \mathbb{F}_p^n$, any $(\{q_j\}_{j\in [\ell]}, {\sf aux}) \leftarrow \mathcal{Q}(n,i)$,   any  answers $\{a_j\}_{j\in [\ell]}$ such that 
        \begin{align*}
            \left|\{ j \in [\ell]: a_j \neq \perp\}\right| &\geq k, \\
            \left|\{  j \in [\ell]:  a_j \not\in \{\perp, \mathcal{A}({\bf x}, q_j)\} \right| &\leq b, 
        \end{align*}
and any $
            {\sf output\_list} \leftarrow \mathcal{R}(i, \{a_j\}_{j \in [\ell]}, {\sf aux}),
        $
it holds that 
        \begin{align}\label{eqn:cts}
         \Pr\big[  ( |{\sf output\_list}| \leq L) \wedge (x_i \in {\sf output\_list})\big]=1.
        \end{align}

\noindent
{\bf $t$-Privacy}.
Informally, the  scheme $\Gamma$ is considered   {\em $t$-private} if any 
collusion of  $\leq t$ servers learns  no information about the client's retrieval index $i$.
Formally, the  scheme $\Gamma$ is {\em $t$-private} if for any $n$, any $i_1, i_2 \in [n]$, and any set $T \subseteq [\ell]$ of size  $ \leq t$, the distributions of $\mathcal{Q}_T(n,i_1)$ and $\mathcal{Q}_T(n,i_2)$ are identical, where $\mathcal{Q}_T$ denotes the concatenation of the $j$-th output of  $\mathcal{Q}$ for all $j\in T$, i.e.,  $\{q_j\}_{j \in T}$

\end{definition}

\noindent{\bf Remark 1.}
Our correctness property requires the  scheme $\Gamma$ to   satisfy 
Eq.  (\ref{eqn:cts}). We call this kind of correctness  {\em perfect}.  
By contrast,  several existing BRPIR schemes \cite{devet2012optimally} may satisfy a relaxed correctness property
that guarantees  
        \begin{align} 
         \Pr\big[  ( |{\sf output\_list}| \leq L) \wedge (x_i \in {\sf output\_list})\big]\geq 1-\epsilon
        \end{align}
        for a very small number $\epsilon$. 
When the failure probability  $\epsilon$ is sufficiently small (e.g., negligible in the number of servers), the scheme is statistically reliable and we call this kind of correctness 
 {\em statistical}. 
 
 \vspace{1mm}
 \noindent{\bf Remark 2.}
 For the ease of exposition, hereafter we denote 
any  $L$-list-decodable   $b$-Byzantine-robust $t$-private $k$-out-of-$\ell$ PIR  scheme
  by \underline{$(L,b,t,k,\ell)$-ldBRPIR}.
By default, when $t=1$ we   simply denote any $(L,b,1,k,\ell)$-ldBRPIR as
$(L,b,k,\ell)$-ldBRPIR and refer to the   property of $1$-privacy as privacy.

    \vspace{1mm}

The efficiency of an $(L,b,t,k,\ell)$-ldBRPIR scheme is mainly measured by its
 communication complexity, which is 
  the  average number of bits that have to be communicated between the client and all servers, in order to retrieve one bit from the database. 
    \begin{definition}[{\bf Communication Complexity}]
    \em
        The {\em communication complexity} of the scheme $\Gamma$, denoted by ${\sf CC}_{\Gamma}(n)$, is defined as the average number of bits communicated per bit retrieved between the client and all servers, maximize over the choices of both the database 
         ${\bf x} \in \mathbb{F}_{p}^{n}$ and the retrieval index $i \in [n]$, i.e., 
        \begin{align*}
            {\sf CC}_{\Gamma}(n) = \max_{{\bf x},i} \left( \frac{1}{\lceil \log_2 p \rceil} \sum_{j=1}^{\ell}(|q_j| + |a_j|) \right).
        \end{align*}
    \end{definition}

\section{Perfect ldBRPIR  based on Overinterpolation}
\label{sec:construction1}
    In this section, we construct a perfectly correct $(L,b,t,k,\ell)$-ldBRPIR scheme $\Gamma_1$ 
  with   list size \(L = k^{O(1)}\) and   
      Byzantine  robustness   \(b \leq k-2\).
      The proposed scheme is most  suitable for a small number of servers and 
      comparably as efficient as 
  a \(t\)-private \((k-b)\)-out-of-\(\ell\) RPIR scheme.

  \begin{figure*}[htbp!]
  
        \begin{center}
		\begin{boxedminipage}{\textwidth}
			\underline{${\sf output\_list} \leftarrow \Gamma_1. \mathcal{R}(i, \{a_j\}_{j\in [\ell]}, {\sf aux})$:}
			\vspace{0.5mm}                
                \begin{itemize}
                    \item[1.]   
                    Parse ${\sf aux} = (\{\lambda_j\}_{j\in [\ell]}, \{{\bf r}_s\}_{s\in [t]})$. Construct the polynomial $G(\lambda) = E(i) + \sum_{s = 1}^{t} \lambda^s {\bf r}_s$.

                    \item[2.]
                    Wlog,  suppose the first $k$ servers respond and 
                     denote   $a_j = (u_j, {\bf v}_j )$ for all $j\in [k]$. 
                    For each $j\in [k]$, let $\alpha_j = u_j, \beta_j = \left< {\bf v}_{j}, G'(\lambda_j) \right> $.              
                    
                    \item[3.] 
                        Initialize a  set ${\sf cp}=\emptyset$ to store the candidate polynomials.  For any  set $H\subseteq [k]$ of cardinality $k-b$, 
                        \begin{itemize}
                            \item[(3.1)] Interpolate a polynomial $f_H(\lambda)$ such that for each $j\in H$, $f_H(\lambda_j) = \alpha_j, f_H'(\lambda_j) = \beta_j$. 

                            \item[(3.2)] If $f_H(\lambda)$ is of degree at most $wt$, add $f_H(\lambda)$ to the set   ${\sf cp}$.
                        \end{itemize}
                        \item[4.] Output ${\sf output\_list} = \{f_H(0):f_H(\lambda)\in {\sf cp}\}$.
                \end{itemize}
		\end{boxedminipage}
	\end{center}
	\caption{Reconstructing algorithm $\Gamma_1.\mathcal{R}$ of the $(L,b,t,k,\ell)$-ldBRPIR scheme $\Gamma_1$.}
	\label{fig:construction1}
\end{figure*}

\subsection{The Construction}
    The scheme $\Gamma_1$ retains the querying and answering algorithms of the Woodruff-Yekhanin RPIR scheme $\Gamma$, $(\Gamma_1.\mathcal{Q} ,\Gamma_1.\mathcal{A}) = (\Gamma.\mathcal{Q}, \Gamma.\mathcal{A})$.  
    The key differences between $\Gamma_1$ and $\Gamma$ are as follows:  
    \textbf{(I)} The degree parameter $w$ is required to satisfy Eq.~\eqref{eq: wy w} in $\Gamma$, whereas in $\Gamma_1$ it is required to satisfy Eq.~\eqref{eqn:w}.  
    \textbf{(II)} The reconstructing algorithm $\Gamma.\mathcal{R}$ in $\Gamma$ is replaced by the reconstructing algorithm $\Gamma_1.\mathcal{R}$ in $\Gamma_1$, as shown in \textbf{Fig}~\ref{fig:construction1}.

The scheme $\Gamma_1$ achieves the expected list size and Byzantine robustness by invoking a novel reconstructing algorithm, which interpolates the polynomial $f(\lambda)=F(G(\lambda))$ of Eq.~(\ref{eqv1}) with more server responses than necessary. More precisely, upon $k$ out of $\ell$ servers responding, the client interpolates $f(\lambda)$ of degree-$wt(<2(k-b)-1)$ with any $k-b$ of the responses, determines whether the interpolated polynomial is a possible candidate of $f(\lambda)$, and includes it into ${\sf output\_list}$ when it is indeed a possible candidate. 
As the crux of this idea, we require 
\begin{align} \label{eqn:deg}
\deg(f(\lambda))<2(k-b)-1
\end{align}
 such that  the $k-b$ evaluations 
and $k-b$ derivatives deduced from the 
$k-b$ server responses are more than necessary and thus enable us to
determine whether each interpolated polynomial is indeed 
a possible candidate with its degree. 
We refer to this technique of using more-than-necessary points to do polynomial interpolation
as 
{\em overinterpolation}. 
In particular, for given $k$ and $b$, the inequality (\ref{eqn:deg}) is met by
choosing a 
  weight 
    parameter $w = O(1)$ in Woodruff-Yekhanin RPIR scheme (see Section \ref{sec:Basic Unique decoding PIR of woodruff and Yekhanin}) such that  
    \begin{align}\label{eqn:w}
    w<\left\lfloor \frac{2(k-b)-1}{t}\right\rfloor. 
    \end{align}
   
    The main observation about our overinterpolation technique includes: (1) 
whenever some of the $k-b$ server responses under consideration 
are incorrect, the interpolated polynomial will be of degree
$>\deg(f(\lambda))$ with very large probability and thus 
be ruled out; and (2) whenever the server responses are all correct, the interpolated 
polynomial will be   of degree $\deg(f(\lambda))$ with probability 1 and thus 
appear in   ${\sf output\_list}$. 
In other words, the degree of $f(\lambda)$ serves as a {\em limit} and 
gives a  filtering process that effectively eliminates numerous erroneous results while preserving the correct one. 
    By adjusting   this  limit, we can regulate the number of polynomials that pass the filtering process.
     When the  limit  is set to an optimal value, the number of polynomials that pass the filtering process
     will be   \( k^{O(1)} \), thereby achieving the expected 
      list size for any admissible    Byzantine  robustness parameter $b$.

%
%

\subsection{Analysis}

In this section, we show that the proposed scheme $\Gamma_1$ is indeed an
$(L,b,t,k,\ell)$-ldBRPIR scheme for  \( L = k^{O(1)} \) and $b\leq k-2$. 

\vspace{2mm}
\noindent
{\bf Correctness.} 
To show that the proposed scheme is correct, we start with a technical lemma that gives an
  upper bound  on the number of the points where two distinct  degree $d$ polynomials have the
    same  order-1 evaluations.
    \begin{lemma} \label{lemma2}
Suppose that $f_1(\lambda)$ and $f_2(\lambda)$ are two  distinct polynomials of degree
$d$ over a finite field $\mathbb{F}_p$. Then there exist at most $\lfloor \frac{d}{2} \rfloor$ field elements $\lambda_j$ such that $f_1(\lambda_j) = f_2(\lambda_j)$ and $f'_1(\lambda_j) = f'_2(\lambda_j)$. 
    \end{lemma}
    \begin{proof}
Assume for contradiction that  there exist $\theta \geq \lfloor \frac{d}{2} \rfloor + 1$ distinct field elements  $\lambda_1, \lambda_2, \dots, \lambda_{\theta}$ such that 
$
 f_1(\lambda_j) = f_2(\lambda_j)$ and $\quad f'_1(\lambda_j) = f'_2(\lambda_j)
$ for all  $j \in [\theta]$. 
Then the nonzero polynomial $g(\lambda) = f_1(\lambda) - f_2(\lambda)$
is of   degree at most $d$ and satisfies 
$$
  \forall j\in [\theta],   ~~     g(\lambda_j) = 0 ~~\text{and} ~~ g'(\lambda_j) = 0.
$$
It follows  that
        \[
            \forall j \in [\theta], ~~ (\lambda - \lambda_j)^2 \mid g(\lambda),
        \]
Hence,  the nonzero polynomial $g(\lambda)$ must be  of degree $\geq 2\theta>d$, which gives a 
    contradiction. \hfill $\square$
    \end{proof}

    Note that the candidate polynomials in the set 
    ${\sf cp}$ generated by our  reconstructing 
     algorithm \( \Gamma_1. \mathcal{R} \) all have degree $\leq 2(k-b)-1$. 
      Lemma \ref{lemma2} shows that any two polynomials 
      from $\sf cp$  cannot have the same order-1 evaluations 
     at too many field elements. In the language of multiplicity codes, the 
     polynomials in $\sf cp$ must be  distant from each other. 
    The following theorem shows that by appropriately choosing
    the parameters $w,b,t$ and $k$, the size $L$ of ${\sf output\_list}$ in our reconstructing algorithm
     can be      made  as small as $k^{O(1)}$.

\begin{theorem} \label{theo:correctness of gamma1}
    The  scheme \( \Gamma_1 \)     is correct  with list size \( L = k^{O(1)} \) when
     \( w < \lfloor \frac{2(k-b) - 1}{t} \rfloor \) and \( wt = O(1) \) with respect to   \( k \).
\end{theorem}

\begin{proof}
Let ${\sf output\_list}$ be the   list output by  $\Gamma_1.\mathcal{R}$. 
As per Eq. (\ref{eqn:cts}), it suffices to show that {\bf (I)} 
{\em $x_i \in {\sf output\_list}$, i.e., the  list output by   $\Gamma_1.\mathcal{R}$ always contains the 
expected block $x_i$};  and {\bf (II)} {\em $|{\sf output\_list}| \leq k^{O(1)}$, i.e.,  the size $L$ of 
the list output by  $\Gamma_1.\mathcal{R}$  is $k^{O(1)}$, under
the proposed choices of the parameters $w,b,t$ and $k$.}

\vspace{2mm}
\noindent
        {\bf (I)} 
Referring to the description of $\Gamma_1.\mathcal{R}$ in Figure \ref{fig:construction1},  wlog the first $k$ servers respond. Consider    the degree-\( wt \) polynomial 
    \[
        f(\lambda) = F(G(\lambda)) = F(E(i) + \sum_{s=1}^{t}  \lambda^{s} {\bf r}_s). 
    \]
    For every $j\in [k]$, it is easy to see that 
    \begin{align*}
        f(\lambda_j) &= F(G(\lambda_j)) = F(q_j), \\ 
        f'(\lambda_j) &= \sum_{c=1}^{m} \left.\frac{\partial F({\bf z})}{\partial z_{c}}\right|_{G(\lambda_j)} \cdot G'(\lambda_j)_c \\
        &= \left< {\bf v}_j  , G'(\lambda_j) \right>, 
    \end{align*}     
    where \( G'(\lambda_j)_c \) denotes the \( c \)-th entry of the vector \( G'(\lambda_j) \). 
   Let $a_j=(\alpha_j,\beta_j)$ be the response of the $j$-th server for every
$j\in[k]$. If at most $b$  of the  responses  \( \{a_j\}_{j \in [k]} \) are incorrect
(where ``the response $a_j$ is incorrect" means that either $\alpha_j\neq F(q_j)$ or 
$\beta_j\neq \langle {\bf v}_j  , G'(\lambda_j) \rangle$),   
    then     for at least \( k-b \) indices \( j \in [k] \) we have  
    $$f^{(\leq 1)}(\lambda_j)=(\alpha_j,\beta_j).$$
When the set $H$ in $\Gamma_1.\mathcal{R}$ is composed of the indices of 
  $k-b$ correct responses, the interpolated polynomial $f_H(\lambda)$
  is exactly equal to $f(\lambda)$ and thus has degree $\leq wt (< 2(k-b) - 1)$ and  
  results in $f_H(0)=f(0)= F(E(i))=x_i$ being included into 
  the set {\sf output\_list}.

\vspace{2mm}
\noindent
        {\bf (II)} 
Referring to the description of $\Gamma_1.\mathcal{R}$ in {\bf Fig.} \ref{fig:construction1},
as ${\sf output\_list} = \{f_H(0) : f_H(\lambda) \in {\sf cp}\}$, it is trivial to see that 
    \begin{equation} \label{eq1}
        |{\sf output\_list}| \leq |{\sf cp}|,
    \end{equation}
where the equality occurs  if and only if all polynomials in 
${\sf cp}$ have different constant terms. 
Let 
    \[
        \mathcal{H} = \{H : H \subseteq [k], |H| = k - b, \deg(f_H(\lambda)) \leq w \}.
    \]
be the set of all   subsets of  $[k]$ of cardinality  $k - b$ that 
gives an interpolated polynomial of degree    $\leq w$. 
According to the construction of $\sf cp$, it is trivial to see that ${\sf cp}=\{f_H(\lambda): H \in \mathcal{H}\}$. 
Therefore,  
 \begin{align}
 |{\sf cp}| \leq |\mathcal{H}|.
 \end{align}
It  is possible that 
 $ |{\sf cp}|< |\mathcal{H}|$ because different subsets $H_1,H_2\in \mathcal{H}$
 may give the same polynomial, i.e., $f_{H_1}(\lambda)=f_{H_2}(\lambda)$. 
Let  $\hat{\mathcal{H}} $ be a subset of $\mathcal{H}$ such that there is a bijection
 $g : \hat{\mathcal{H}} \to {\sf cp}$. Then 
    \[
        |{\sf cp}| = |\hat{\mathcal{H}}|.
    \]  
    For any two distinct subsets $H_1, H_2 \in \hat{\mathcal{H}}$, 
    the degree-$wt$  polynomials $g(H_1) = f_{H_1}(\lambda)$ and $g(H_2) = f_{H_2}(\lambda)$
     must be distinct. By Lemma \ref{lemma2}, the number of indices  
      $j\in[k]$ such that  $f_{H_1}^{(\leq 1)}(\lambda_j) = f_{H_2}^{(\leq 1)}(\lambda_j)$ 
       is upper bounded by $\lfloor wt/2 \rfloor$. Therefore,  we must have that  
  $
        | H_1 \cap H_2 | \leq \lfloor wt/2 \rfloor,
 $ and thus 
\begin{align}\label{eqn:H12}
 |H_1 \setminus H_2| = |H_2 \setminus H_1| \geq k-b- \lfloor wt/2 \rfloor.
\end{align}
     For every $H \in \hat{\mathcal{H}}$, we define a   binary vector $c_H = (c_{H,1}, \dots, c_{H,k})$ of length $k$ such that 
    \[
        c_{H,j} = \begin{cases}
            1, & j \in H, \\
            0, & j \notin H.
        \end{cases}
    \]  
    Then it is easy to see that $c_{H}$ is of Hamming weight $k-b$. 
    As per Eq.  (\ref{eqn:H12}), the Hamming distance between 
     $c_{H_1}$ and $c_{H_2}$ for any two distinct subsets $H_1,H_2\in \hat{\mathcal{H}}$ is at least $2\delta$ for 
    $$\delta = k-b-\lfloor wt/2 \rfloor.$$
Therefore, $C=\{c_H\}_{H \in \hat{\mathcal{H}}}$ is a binary $(k, 2\delta, k-b)$ constant weight code. 
Let $A(k, 2\delta, k-b)$ be the maximum size of a binary $(k, 2\delta, k-b)$ constant weight codes. 
As per Theorem \ref{theo:constant-weight-code}, we have that 
    \begin{align*}        
        |\hat{\mathcal{H}}| = |\{c_{H}\}_{H\in \hat{\mathcal{H}}}| 
        &\leq A(k,2\delta, k-b) \\ 
        &\leq \frac{{k \choose k-b-\delta+1}}{{k-b \choose k-b-\delta + 1}}\\
        &= \frac{{k \choose \lfloor wt/2 \rfloor+1}}{{k-b \choose \lfloor wt/2 \rfloor+1}}\\
        &= O((\frac{k}{k-b})^{\lfloor wt/2 \rfloor+1}) \\
        &= k^{O(1)},
    \end{align*}
  i.e.,  the size $L$ of ${\sf output\_list}$ is at most $k^{O(1)}$.       \hfill $\square$
\end{proof}

\vspace{2mm}
\noindent
{\bf $t$-Privacy.}
The $t$-privacy property of the proposed scheme is identical to that of 
the Woodruff-Yekhanin RPIR scheme, because both schemes share the same 
querying algorithm. 
}

    \vspace{1mm}
    \noindent
    {\bf Communication complexity.}
    The querying algorithm requires  the client to send a vector \(q_j\in \mathbb{F}_p^m\) to each of the $\ell$ servers.
    The answering algorithm requires each server \(\mathcal{S}_j\) to send a vector 
    \(a_j = (u_j, {\bf v}_j)\in \mathbb{F}_p^{m+1}\)
    to the client. Therefore, the average number of bits exchanged for retrieving one bit from the database is \(\ell(2m+1)=O(\ell m)\).
    As the integer  \(m\) is chosen such that 
    $
        {m \choose w} \geq n,
    $
    we have that \(m = O(w n^{1/w})\) and thus 
    \[
        CC_{\Gamma_1}(n) = O(\ell w n^{1/w}).
    \] 
    Note that $n\gg w$ in a typical application scenario of PIR, therefore $CC_{\Gamma_1}(n)$ is a monotonically decreasing function of \(w\). Ideally, we prefer to choose a large \( w \) in order to reduce the communication complexity.
    Theorem \ref{theo:correctness of gamma1} stipulates that the list size \( L = \left( \frac{k}{k-b} \right)^{\lfloor wt/2 \rfloor + 1} \).
    We present different strategies for choosing \( w \), which correspond to different list sizes and communication complexities.
    \begin{itemize}
        \item 
        When \( \frac{k}{k-b} = O(1) \), that is, there exists a constant integer \( c \) such that \( b < k - k/c \), we can choose \( w = O\left( \frac{\log k}{t} \right) \). 
        In this case, the communication complexity of our scheme \( \Gamma_1 \) is given by \( CC_{\Gamma_1}(n) = \frac{\ell \log k}{t} n^{O(t / \log k)} \), and the list size is \( L = k^{O(1)} \). 
    
        \item 
        When \( \frac{k}{k-b} = O(k^c) \) for some constant \( c \leq 1 \), we can choose \( w = O\left( \frac{1}{t} \right) \). 
        In this case, the communication complexity of our scheme \( \Gamma_1 \) is given by \( CC_{\Gamma_1}(n) = \frac{\ell}{t} n^{1 / O(1/t)} \), and the list size is \( L = k^{O(1)} \). 
    \end{itemize}
    Notably, when \( k \leq 2^5 \), the differences between \( O(k/t) \), \( O(\log k / t) \), and \( O(1 / t) \) are not significant. 
    In this case, we can directly set \( w = \left\lfloor \frac{2k-2b-2}{t} \right\rfloor \).  
    The communication complexity is 
    \[
        CC_{\Gamma_1}(n) = O\left( \ell (k-b) n^{1 / \lfloor (2k-2b-2) / t \rfloor} \right), 
    \] 
    and the list size is \( L = \left( \frac{k}{k-b} \right)^{k-b} \). 
    For the ease of sublinear communication complexity, we require that $b \leq k-2$.

\subsection{Improving the Client-Side Computation Complexity} \label{sec:Improving the client-side}   
    In Step $3$ of reconstructing algorithm  $\Gamma_1.\mathcal{R}$, the client interpolates a polynomial with the set $H$ for any $H\subset [k]$ of cardinality $k-b$, which incurs the $\binom{k}{k-b}$ interpolations. 
    In this section, we propose a method to reduce the number of interpolations to \( \binom{k}{wt/2} \).
    Since \( wt < 2(k-b) - 1 \) and \( wt = O(1) \) with respect to the parameter \( k \), the client-side computational complexity is significantly reduced. Specifically, the complexity decreases from exponential in \( k \) to a polynomial function of \( k \).

    We replace the Step $3$ of $\Gamma_1.\mathcal{R}$ with the following step:
    \begin{itemize}
        \item[3.] 
        Initialize a  set ${\sf cp}=\emptyset$ to store the candidate polynomials.  For any set $H\subseteq [k]$ of cardinality $\lfloor wt/2 \rfloor + 1$, 
        \begin{itemize}
            \item[(3.1)] Interpolate a polynomial $f_H(\lambda)$ such that for each $j\in H$, $f_H(\lambda_j) = \alpha_j, f_H'(\lambda_j) = \beta_j$. 

            \item[(3.2)] If $f_H(\lambda)$ is of degree at most $wt$ and there exist at least $k-b$ distinct values of $j \in [k]$ satisfying $f(\lambda_j) = \alpha_j, f'(\lambda_j) = \beta_j$  add $f_H(\lambda)$ to the set   ${\sf cp}$.
        \end{itemize}
    \end{itemize}
    
    To see why this substitution works, consider any polynomial \( f_H(\lambda) \in {\sf cp} \) constructed in the original Step 3 of \( \Gamma_2.\mathcal{R} \) for a set \( H \subseteq [k] \) of cardinality \( k-b \). Since \( f_H(\lambda) \) is included in \( {\sf cp} \), its degree must not exceed \( wt \). 
    Now, consider any subset \( H' \subseteq H \) with cardinality \( \lfloor wt/2 \rfloor + 1 \). The polynomial \( f_{H'}(\lambda) \), which satisfies \( f_{H'}^{(\leq 1)}(\lambda_j) = (\alpha_j, \beta_j) \) and has degree at most \( wt+1 \), is uniquely determined as \( f_{H'}(\lambda) = f_H(\lambda) \). 
    The proof can be established by contradiction and is omitted here for brevity. 
    Consequently, the polynomial \( f_H(\lambda) \) will also be included in the set \( {\sf cp} \) under the new Step 3.

    With the substitution of Step 3, we can reduce the client-side computation complexity of $\Gamma_1.\mathcal{R}$ to be polynomial in $k$.

{
\section{Perfect ldBRPIR based on weighted-degree polynomial} \label{sec:construction2}
     
        In this section, we construct a perfectly correct $(L,b,t,k,\ell)$-ldBRPIR scheme $\Gamma_2$ 
  with   list size \(L =O(k)\) and   
      Byzantine  robustness   \(b \leq k-\sqrt{kt}\).
     Compared with $\Gamma_1$, the scheme $\Gamma_2$ demonstrates reduced communication complexity
       when dealing with a larger number of servers.

     \begin{figure*}[t]
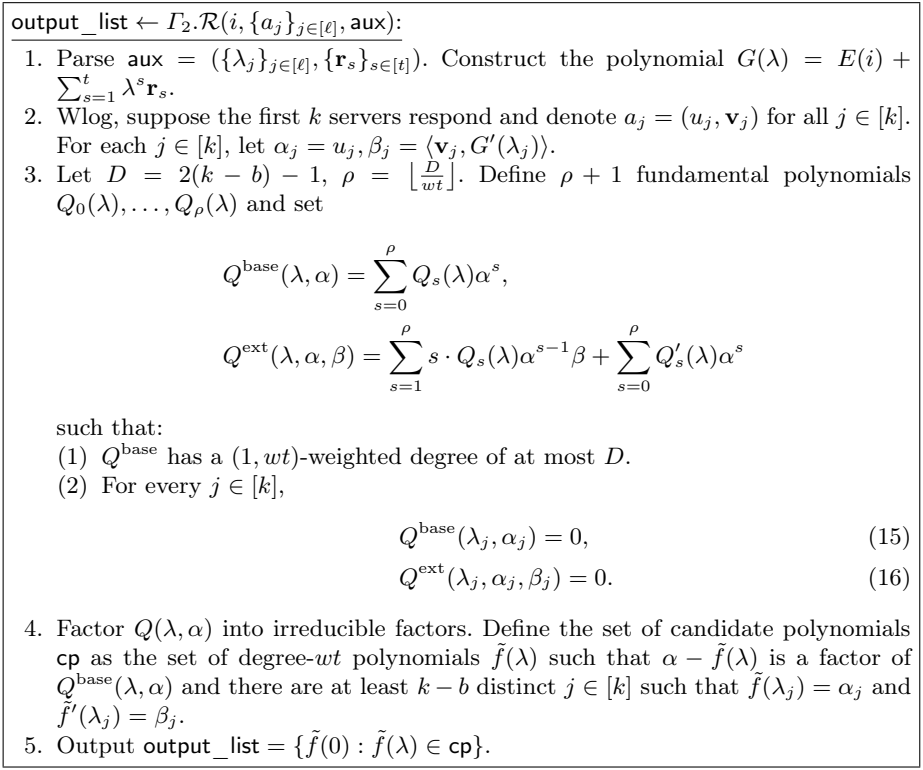

    	\begin{center}
    		\begin{boxedminipage}{\textwidth}
			\underline{${\sf output\_list} \leftarrow \Gamma_2. \mathcal{R}(i, \{a_j\}_{j\in [\ell]}, {\sf aux})$:}
			\vspace{0.5mm}
                    \begin{itemize}
                        \item[1.]   
                        Parse ${\sf aux} = (\{\lambda_j\}_{j\in [\ell]}, \{{\bf r}_s\}_{s\in [t]})$. Construct the polynomial $G(\lambda) = E(i) + \sum_{s = 1}^{t} \lambda^s {\bf r}_s$.

                        \item[2.]
                        Wlog,  suppose the first $k$ servers respond and denote   $a_j = (u_j, {\bf v}_j )$ for all $j\in [k]$. 
                        For each $j\in [k]$, let $\alpha_j = u_j, \beta_j = \left< {\bf v}_{j}, G'(\lambda_j) \right> $.
                        
                        \item[3.] 
                        Let \( D = 2(k-b) - 1 \), \( \rho = \left\lfloor \frac{D}{wt} \right\rfloor \). Define \(\rho + 1\) fundamental polynomials \( Q_0(\lambda), \ldots, Q_{\rho}(\lambda) \) and set
                        \begin{align*}                        
                            &Q^{\text{base}}(\lambda,\alpha) = \sum_{s=0}^{\rho} Q_{s}(\lambda)\alpha^s, \\
                            &Q^{\text{ext}}(\lambda,\alpha,\beta) = \sum_{s=1}^{\rho} s \cdot Q_s(\lambda)\alpha^{s-1}\beta + \sum_{s=0}^{\rho} Q'_s(\lambda)\alpha^s
                        \end{align*}
                        such that:
                        \begin{itemize}
                            \item[(1)] 
                            \( Q^{\text{base}} \) has a \( (1, wt) \)-weighted degree of at most \( D \). 
                            \item[(2)] For every \( j \in [k] \), 
                            \begin{align}
                                &Q^{\text{base}}(\lambda_j, \alpha_j) = 0, \label{eq1} \\ 
                                &Q^{\text{ext}}(\lambda_j, \alpha_j, \beta_j) = 0. \label{eq2}
                            \end{align}
                        \end{itemize}

                        \item[4.] 
                        Factor $Q(\lambda, \alpha)$ into irreducible factors.
                        Define the set of candidate polynomials \({\sf cp}\) as the set of degree-\(wt\) polynomials \(\tilde{f}(\lambda)\) such that \(\alpha - \tilde{f}(\lambda)\) is a factor of \(Q^{\text{base}}(\lambda, \alpha)\) and there are at least $k-b$ distinct $j\in[k]$ such that \(\tilde{f}(\lambda_j) = \alpha_j\) and $\tilde{f}'(\lambda_j) = \beta_j$.

                        \item[5.] 
                        Output ${\sf output\_list} = \{\tilde{f}(0): \tilde{f}(\lambda)\in {\sf cp}\}$.
                    \end{itemize}
    		\end{boxedminipage}
    	\end{center}
    	\caption{Reconstructing algorithm $\Gamma_2.\mathcal{R}$ of the $(L, t, k, \ell, b)$-ldBRPIR scheme $\Gamma_2$.}
    	\label{fig:construction2}
    \end{figure*}

\subsection{The Construction}
{
    Same as $\Gamma_1$, the scheme $\Gamma_2$ retains the querying and answering algorithms of the Woodruff-Yekhanin RPIR scheme $\Gamma$, $(\Gamma_2.\mathcal{Q} ,\Gamma_2.\mathcal{A}) = (\Gamma.\mathcal{Q}, \Gamma.\mathcal{A})$.  
    The key differences between $\Gamma_2$ and $\Gamma$ are as follows:  
    \textbf{(I)} The degree parameter $w$ is required to satisfy Eq.~\eqref{eq: wy w} in $\Gamma$, whereas in $\Gamma_2$ it is required to satisfy that $w \leq (k-b)^2 / t$.  
    \textbf{(II)} The reconstructing algorithm $\Gamma.\mathcal{R}$ in $\Gamma$ is replaced by the reconstructing algorithm $\Gamma_2.\mathcal{R}$ in $\Gamma_2$, as shown in \textbf{Fig}~\ref{fig:construction2}.
}
    
    In scheme $\Gamma_2$, we achieve the expected list size and Byzantine robustness by extending Sudan's list-decoding algorithm for Reed-Solomon codes to the order-$1$ univariate multiplicity codes.
Each responding server $\mathcal{S}_j (j\in[k])$ in the proposed scheme $\Gamma_2$ returns both an evaluation of the polynomial $F$ that encodes the database $\bf x$ and $m$ partial derivatives of $F$, all 
computed at the same query
 point $q_j\in \mathbb{F}_p^m$. The response of each server $\mathcal{S}_j$ gives an order-1 evaluation
 $f(\lambda)=F(G(\lambda))$ at a field element $\lambda_j$, i.e.,
  $f^{(\leq 1)}(\lambda_j)$.  When at most 
  $b$ out of the $k$ responding servers are malicious and respond incorrectly, the problem of
constructing a list that contains $x_i=f(0)$ can be reduced to the problem of
building a list that contains  $f(\lambda)$, which is exactly the problem of list decoding
an order-1 univariate multiplicity code for $f(\lambda)$ from the (corrupted) codeword
$\{(\alpha_j, \beta_j)\}_{j=1}^k$, where 
$(\alpha_j,\beta_j)=f^{(\leq 1)}(\lambda_j)$  
for any honest server $\mathcal{S}_j$ but may not be true for a malicious server. 
 
 While Sudan's list decoding algorithm \cite{sudan1997decoding} enables one to correct a relatively large fraction of errors
 but is only applicable to order-0 univariate multiplicity codes (i.e., Reed-Solomon codes) and 
 the recent list decoding algorithms \cite{guruswami2011optimal} are applicable to general order-$s$ univariate 
 multiplicity codes but only allow recovery from a relatively small fraction of errors,
 the main innovative idea underlying  $\Gamma_2$ is achieving the 
 (possibly) best from both worlds   by  injecting Sudan's idea of using 
the non-linear terms $\alpha^c$ in the  weighted-degree  bivariate polynomial $Q$ of Eq. (\ref{eqn:sudan})
 into the construction of the multivariate polynomial $Q$ of Eq. (\ref{eqn:gw}), which is 
 linear in  $\alpha_0,\ldots,\alpha_s$. 
 To realize this idea, we simply  replace  the linear term
 $\alpha_c$ in  Eq. (\ref{eqn:gw}) with its properly selected   higher degree powers. 
The main observation about this simple modification to 
Eq. (\ref{eqn:gw}) is that the new polynomial may  accommodate a much larger  number of monomials and 
thus significantly improve the fraction of correctable errors. 

We denote by  \( Q^{\text{base}} (\lambda, \alpha)\)   the new polynomial resulted from the 
aforementioned idea and denote by    $Q^{\text{ext}}(\lambda, \alpha, \beta)$ 
the polynomial obtained by  applying the special operator $\mathcal{D}$ from Section \ref{sec:umc}
to $Q^{\text{base}} (\lambda, \alpha)$, where $\beta=\mathcal{D}(\alpha)$. 
More precisely, for $\rho = \left\lfloor \frac{D}{wt} \right\rfloor$, we 
choose a bivariate base polynomial 
\begin{align}\label{eqn:Qbase}
        Q^{\text{base}}(\lambda, \alpha) = \sum_{s=0}^\rho Q_s(\lambda)\alpha^s.
\end{align}
and impose   $k$ constraints of the form  
$$ Q^{\text{base}}(\lambda_j, \alpha_j) = 0 $$ for all \( j \in [k] \). 
We calculate the trivariate extended polynomial 
$$
Q^{\text{ext}}(\lambda,\alpha,\beta)=\mathcal{D}(Q^{\text{base}}) = \sum_{s=1}^{\rho} s \cdot Q_s(\lambda)\alpha^{s-1}\beta + \sum_{s=0}^{\rho} Q'_s(\lambda)\alpha^s$$
and impose     \( k \)  additional constraints  of the form
    $$
        Q^{\text{ext}}(\lambda_j, \alpha_j, \beta_j) = 0 
    $$ 
    for all \( j \in [k] \).
    
    To better understand the roles of the  polynomials $\{Q_s(\lambda)\}_{s=0}^\rho$, consider the polynomial $p(\lambda) = Q^{\text{base}}(\lambda, \tilde{f}(\lambda))$, where   $\tilde{f}(\lambda)$ is a polynomial  of degree $\leq wt$ such that    $\tilde{f}^{(\leq 1)}(\lambda_j) = (\alpha_j,\beta_j)$
     for at least $k-b$ distinct indices $j\in [k]$. 
If $j\in[k]$ is any index  such that  $\tilde{f}^{(\leq 1)}(\lambda_j) = (\alpha_j,\beta_j)$, 
then we would have  
\begin{align*}
        p(\lambda_j) &= Q^{\text{base}}(\lambda_j, \tilde{f}(\lambda_j)) = 0; \\
         p'(\lambda_j) &= Q^{\text{ext}}(\lambda_j, \tilde{f}(\lambda_j), \tilde{f}'(\lambda_j)) = 0.
\end{align*} 
    Hence, the polynomial $p(\lambda)$ should be of degree  $\geq 2(k-b)$ if it is nonzero.
By choosing a weighted degree parameter   $D =2(k-b)-1$ and properly choosing  
degrees for the polynomials  $\{Q_s(\lambda)\}_{s=0}^\rho$, we force  
 $p(\lambda)$ to be   a polynomial of degree   $\leq D$ and thus  
  identically zero, which in turn enable use to extract  $\tilde{f}(\lambda)$
by factoring the  interpolated polynomial $Q^{\text{base}}(\lambda, \alpha)$.

    

\subsection{Analysis}
    In this section, we show that the proposed scheme $\Gamma_2$ is indeed an $(L,b,t,k,\ell)$-ldBRPIR scheme for  \( L = O(k) \) and $b\leq k-\sqrt{kt}$. 

    \vspace{2mm}
    \noindent
    {\bf Correctness.} 
    To show that the proposed scheme is correct, we 
 firstly identify under what conditions there do exist  $\rho+1$ polynomials 
 $\{Q_s(\lambda)\}_{s=0}^{\rho}$ that satisfy the 
 the constraints of Eq. (\ref{eq1}) and (\ref{eq2}). Such conditions 
 may guide our selections of parameters. 
     \begin{lemma}
    \label{claim:1}
 If $\rho = \lfloor D/wt \rfloor$ and $w \leq \frac{(D+1)^2}{4kt} \left(=\frac{(k-b)^2}{kt}\right)$, there exist $\rho + 1$ polynomials $\{Q_s(\lambda)\}_{s=0}^{\rho}$ that satisfy the constraints imposed by
 Eq. (\ref{eq1}) and (\ref{eq2}), where $\deg(Q_s)\leq D-swt$ 
 for all $s=0,1,\ldots,\rho$.  
    \end{lemma} 
    \begin{proof}
        Consider the bivariate polynomial \(Q^{\text{base}}(\lambda, \alpha)\) of
         Eq. (\ref{eqn:Qbase}). 
         If we choose  \( D=2(k-b)-1\) and choose every  $Q_s(\lambda)$ there to be a 
         polynomial of degree $\leq D-swt$, then the  
          \((1, wt)\)-weighted degree  of \(Q^{\text{base}}(\lambda, \alpha)\)  
       is $\leq D$. 
          Let \({\sf num}\) be the number   of monomials in   \(Q^{\text{base}}(\lambda, \alpha)\).
           Given that $\rho = \lfloor D/wt \rfloor$,  we have 
        \[
            {\sf num} = \sum_{s=0}^{\rho} (D - swt + 1) \geq \frac{(D+1)^2}{2wt}+1.
        \]
           On the other hand, the  number of constraints imposed by Eq. (\ref{eq1}) and (\ref{eq2}) is \(2k\).
When  $w \leq \frac{(D+1)^2}{4kt}$, we have that 
        $$
        2k \leq \frac{(D+1)^2}{2wt} < {\sf num}.
        $$
As the number of coefficients in $Q^{\text{base}}(\lambda, \alpha)$ exceeds the number of constraints given by Eq. (\ref{eq1}) and (\ref{eq2}), there must exist  a nonzero bivariate polynomial $Q^{\text{base}}(\lambda, \alpha)$ that satisfies all of the $2k$ constraints.     The existence of  \(Q^{\text{base}}(\lambda, \alpha)\) implies that of  \(Q_0(\lambda), \ldots, Q_{\rho}(\lambda)\).
        \hfill $\square$    
    \end{proof}

    Lemma \ref{claim:1} shows the   existence of the polynomials $Q_0,..., Q_{\rho}$  when we properly choose the parameters    in $\Gamma_2$.
    Since \( w \) is the degree of the polynomial \( F({\bf z}) \) encoding the database \( {\bf x} \), we must have that  \( w \geq 1 \), which together with the condition
    $w \leq   (k-b)^2/(kt) $ implies that    \( b \leq k - \sqrt{kt} \), i.e.,
    the scheme $\Gamma_2$ can tolerate as many as $k - \sqrt{kt}$ malicious servers that respond incorrectly. 
     \begin{theorem} \label{theo:correctness of gamma2}
        The   scheme $\Gamma_2$   is correct with list size $L= O(k)$ when $1 \leq  w \leq \frac{(k-b)^2}{kt}$.
    \end{theorem}

    \begin{proof}
        Let ${\sf output\_list}$ be the list output by $\Gamma_2.\mathcal{R}$. 
It suffices to show that 
        {\bf (I)} {\em $x_i \in {\sf output\_list}$}; and 
        {\bf (II)} {\em $|{\sf output\_list}| \leq O(k)$, under the proposed choices of the parameters $w,b,t$ and $k$.}

\vspace{2mm}
\noindent
        {\bf (I)} 
        Referring to the description of $\Gamma_2.\mathcal{R}$ in Figure \ref{fig:construction2}, wlog the first $k$ servers respond. Consider    the degree-\( wt \) polynomial 
        \[
            f(\lambda) = F(G(\lambda)) = F(E(i) + \sum_{s=1}^{t}  \lambda^{s} {\bf r}_s). 
        \]
        For every $j\in [k]$, it is easy to see that 
        \begin{align*}
            f(\lambda_j) &= F(G(\lambda_j)) = F(q_j), \\ 
            f'(\lambda_j) &= \sum_{c=1}^{m} \left.\frac{\partial F({\bf z})}{\partial z_{c}}\right|_{G(\lambda_j)} \cdot G'(\lambda_j)_c \\
            &= \left< {\bf v}_j  , G'(\lambda_j) \right>, 
        \end{align*}     
        where \( G'(\lambda_j)_c \) denotes the \( c \)-th entry of the vector \( G'(\lambda_j) \). 
        Consider the polynomial \( p(\lambda) = Q^{\text{base}}(\lambda, f(\lambda)) \).  
        Since $f(\lambda)$ is a polynomial of degree $wt$ and the bivariate polynomial $Q^{\text{base}}(\lambda, \alpha)$ has a $(1,wt)$-weighted degree  $\leq D$, the degree of $p(\lambda)$ is  $\leq D$.
        For any $j\in [k]$ and any triplet $(\lambda_j, \alpha_j, \beta_j)$ such that
        \( f^{(\leq 1)}(\lambda_j) = (\alpha_j,\beta_j) \), we have
        \begin{align*}
            p(\lambda_j) 
            &= Q^{\text{base}}(\lambda_j, f(\lambda_j)) \\
            &= Q^{\text{base}}(\lambda_j, \alpha_j) \\
            &= 0,\\
            p'(\lambda_j) 
            &= \sum_{s=1}^{\rho} s\cdot Q_{s}(\lambda_j) f(\lambda_j)^{s-1} f'(\lambda_j) + \sum_{s=0}^{\rho} Q'_{s}(\lambda_j) f(\lambda_j)^s \\
            &= Q^{\text{ext}}(\lambda_j, f(\lambda_j), f'(\lambda_j)) \\
            &= Q^{\text{ext}}(\lambda_j, \alpha_j, \beta_j) \\
            &= 0.
        \end{align*}
Thus, the polynomial $p(\lambda)$ should be of degree  $\geq 2(k-b)$, if it is nonzero.
On the other hand, our choices of all parameters imply that
 \( \deg(p(\lambda))\leq D=2(k-b) - 1 \). Hence,  \( p(\lambda) \) must be the zero polynomial, i.e., 
  \( Q^{\text{base}}(\lambda, f(\lambda)) = 0 \).  It follows  that \( (\alpha - f(\lambda)) \mid Q^{\text{base}}(\lambda, \alpha) \)
Thus, $f(\lambda)$ will be included in the set  ${\sf cp}$ of candidate polynomials can cause  $f(0) = x_i$ to 
be included into the set ${\sf output\_list}$.

\vspace{2mm}
\noindent
        {\bf (II)}    
  Referring to the description of $\Gamma_2.\mathcal{R}$ in Figure \ref{fig:construction2},  for every
  $\tilde{f}(\lambda)\in {\sf cp}$, the polynomial $\alpha-\tilde{f}(\lambda)$ must be a linear factor of 
         $Q^{\text{base}}(\lambda, \alpha)$.
         As         $Q^{\text{base}}(\lambda, \alpha)$ is a polynomial of degree $\leq \rho$
          in
          $\alpha$, the number of  polynomials in ${\sf cp}$ must be at most 
          $\rho$, i.e., $|{\sf cp}|\leq \rho$. 
                 Consequently,  we have that 
        \[
            L =|{\sf output\_list}|\leq |{\sf cp}|\leq  \rho = \lfloor \frac{2(k-b)-1}{wt} \rfloor = O(k). 
        \]    
In particular, we have that $L \leq 2k$ for any choices of $(k,b,t)$. 
Furthermore, when the majority of the server responses are incorrect, that is, $b > k/2$, we have 
        $
            L \leq k.
        $
        \hfill $\square$
    \end{proof}
}

    \vspace{1mm}
    \noindent
    {\bf $t$-Privacy.}
    The $t$-privacy property of the proposed scheme is identical to that of the Woodruff-Yekhanin RPIR scheme and our scheme $\Gamma_1$, because those schemes share the same querying algorithm. 


    \vspace{1mm}
    \noindent
    {\bf Communication complexity.}
    Same as scheme $\Gamma_1$,
    we have \(m = O(w n^{1/w})\) in scheme $\Gamma_2$ and the communication complexity is 
    \[
        CC_{\Gamma_2}(n) = O(\ell w n^{1/w}).
    \] 
    Note that $n\gg w$ in a typical application scenario of PIR, therefore $CC_{\Gamma_2}(n)$ is a monotonically decreasing function of \(w\). Ideally, we prefer to choose a large \( w \) in order to reduce the communication complexity.
    Theorem \ref{theo:correctness of gamma2} establishes a scaled upper bound for \(w\), specifically \(w \leq \frac{(k-b)^2}{kt}\), implying that the communication complexity \( CC_{\Gamma_2}(n) \) of our scheme \(\Gamma_2\) is
    \[
        CC_{\Gamma_2}(n) =  O\left(\frac{\ell(k-b)^2}{kt} n^{1/ \lfloor \frac{(k-b)^2}{kt} \rfloor}\right).
    \]

    Our scheme \( \Gamma_2 \) shares the same Byzantine tolerance bound \( b \leq k - \sqrt{kt} \) as the scheme in \cite{goldberg2007improving}. 
    For different numbers of malicious servers \( b \), the relationship between the communication complexity $CC_{\Gamma_2}(n)$ of scheme \( \Gamma_2 \) and the database size \( n \) is compared with the schemes in \cite{goldberg2007improving} and \cite{devet2012optimally} as follows:  
    \begin{itemize}  
        \item  
        When \( k - \sqrt{2kt} < b \leq k - \sqrt{kt} \), the communication complexity of our scheme \( \Gamma_2 \) is \( O(n) \), which is less efficient than the communication complexity of the schemes in both \cite{goldberg2007improving} and \cite{devet2012optimally}, achieving \( O(n^{1/2}) \).  

        \item  
        When \( k - \sqrt{3kt} < b \leq k - \sqrt{2kt} \), the communication complexity of our scheme \( \Gamma_2 \) is \( O(n^{1/2}) \), matching the efficiency of the schemes in \cite{goldberg2007improving} and \cite{devet2012optimally}.  

        \item  
        When \( b \leq k - \sqrt{3kt} \), the communication complexity of our scheme \( \Gamma_2 \) can achieve \( o(n^{1/2}) \), offering a significant improvement over the schemes in \cite{goldberg2007improving} and \cite{devet2012optimally}.  
    \end{itemize}

\section{Implementation} \label{sec:implementation}

We implemented the scheme $\Gamma_1$ and $\Gamma_2$ presented in this paper using C++ and the FLINT library \cite{flint} and compared the efficiency of the scheme described in \cite{goldberg2007improving} and \cite{devet2012optimally}. 
The database used in the experiments was generated with a custom Python-based generator.
Our implementation of the schemes and the database generator is publicly available on \href{https://github.com/LDPIR/Efficient-list-decodable-BRPIR-with-Higher-Byzantine-Tolerance}{GitHub}. We have uploaded only a smaller example database and Larger databases can be reproduced using the generator provided in our repository.

We measure the performance of our scheme on a computer with a 16-core intel Xeon Gold 6250
CPU \@ 3.90GHz and 256 GB RAM, running Ubuntu 20.04. All of our experiments are single-threaded. 

In the foregoing sections, we scrutinized the Byzantine robustness of $\Gamma_1$ and $\Gamma_2$. In our experiments, we focus on evaluating 
\textbf{(1)} 
the communication volume exchanged between the client and the servers; and
\textbf{(2)} 
the size of the output lists generated by the schemes.

{\bf (1)}
To evaluate the communication volume between the client and servers, we analyze the total size of the queries \(\{q_j\}_{j \in [\ell]}\) generated by the querying algorithm and the responses \(\{a_j\}_{j \in [\ell]}\) produced by the answering algorithm.

{\bf (2)}
To determine the output list size, we count the number of elements in the list \({\sf output\_list}\) produced by the reconstructing algorithm. 
For different databases \({\bf x}\), indices \(i\) of interest to the client, chosen queries \(\{q_j\}_{j\in [\ell]}\), and adversarial responses \(\{\hat{a}_j\}_{j\in [b]}\) from Byzantine servers, the size of the output list from the reconstructing algorithm varies. Moreover, it is challenging to identify a specific combination of \({\bf x}\), \(i\), \(\{q_j\}_{j\in [\ell]}\), and \(\{\hat{a}_j\}_{j\in [b]}\) that maximizes the output list size.
To address this, we fix modest values for the database size \(n\), the field size \(p\), and the parameters \((k,b,t)\), then allow the adversary to inject random erroneous responses. We execute the protocol repeatedly and record the worst‐case list size yielded by the reconstructing algorithm.

\clearpage

\begin{table}[h!]
\centering
\small
\begin{tabular}{c|c@{\hskip 4pt}c@{\hskip 4pt}c|c@{\hskip 4pt}c@{\hskip 4pt}c}

\toprule
\multicolumn{1}{c|}{\multirow{2}{*}{\textbf{Scheme}}}
  & \multicolumn{3}{c|}{\bf Vary $k$}
  & \multicolumn{3}{c}{\bf Vary $b$} \\
\cmidrule(lr){2-4} \cmidrule(lr){5-7}
  & $k$ & Asymptotic & Experimental
  & $b$ & Asymptotic & Experimental \\

  
\midrule

\multirow{5}{*}{$\Gamma_1$}
  & $16$  & $O(n^{1/4})$ & $20.0 \pm 0.1$ MiB
  & $10$  & $O(n^{1/16})$ & $735.5 \pm 20 $ KiB\\
  & $18$  & $O(n^{1/8})$ & $2.2 \pm 0.1$ MiB 
  & $11$  & $O(n^{1/14})$ & $821.3 \pm 20 $ KiB \\
  & $20$  & $O(n^{1/12})$ & $1.0 \pm 0.1$ MiB
  & $12$  & $O(n^{1/12})$ & $1.0 \pm 0.1$ MiB \\
  & $22$  & $O(n^{1/16})$ & $732.2 \pm 20 $ KiB
  & $13$  & $O(n^{1/10})$ & $1.4 \pm 0.1$ MiB \\
  & $24$  & $O(n^{1/20})$ & $612.5 \pm 20 $ KiB
  & $14$  & $O(n^{1/8})$ & $2.2 \pm 0.1$ MiB \\
\midrule
\multirow{5}{*}{$\Gamma_2$}
  & $16$ & $O(n)$    & \text{Too large}
  & $10$ & $O(n^{1/6})$ & $1.1 \pm 0.1$ MiB \\
  & $18$ & $O(n^{1/2})$ & $325.8 \pm 5$ MiB
  & $11$ & $O(n^{1/5})$ & $1.5 \pm 0.1$ MiB \\
  & $20$ & $O(n^{1/4})$ & $3.6 \pm 0.1$ MiB
  & $12$ & $O(n^{1/4})$ & $3.6 \pm 0.1$ MiB \\
  & $22$ & $O(n^{1/5})$ & $1.5 \pm 0.1$ MiB
  & $13$ & $O(n^{1/3})$ & $15.3 \pm 0.2$ MiB\\
  & $24$ & $O(n^{1/8})$ & $421.5 \pm 20$ KiB
  & $14$ & $O(n^{1/2})$ & $317.2 \pm 5$ MiB \\
\midrule
\multirow{5}{*}{\cite{goldberg2007improving}}
  & $16$ & /            &  /
  & $10$ & $O(n^{1/2})$ & $12.8\pm0.1$ MiB \\
  & $18$ & $O(n^{1/2})$ & $12.8\pm0.1$ MiB
  & $11$ & $O(n^{1/2})$ & $12.8\pm0.1$ MiB \\
  & $20$ & $O(n^{1/2})$ & $12.8\pm0.1$ MiB
  & $12$ & $O(n^{1/2})$ & $12.8\pm0.1$ MiB \\
  & $22$ & $O(n^{1/2})$ & $12.8\pm0.1$ MiB
  & $13$ & $O(n^{1/2})$ & $12.8\pm0.1$ MiB \\
  & $24$ & $O(n^{1/2})$ & $12.8\pm0.1$ MiB
  & $14$ & $O(n^{1/2})$ & $12.8\pm0.1$ MiB \\
\midrule
\multirow{5}{*}{\cite{devet2012optimally}}
  & $16$ & $O(n^{1/2})$ & $12.8\pm0.1$ MiB
  & $10$ & $O(n^{1/2})$ & $12.8\pm0.1$ MiB \\
  & $18$ & $O(n^{1/2})$ & $12.8\pm0.1$ MiB
  & $11$ & $O(n^{1/2})$ & $12.8\pm0.1$ MiB \\
  & $20$ & $O(n^{1/2})$ & $12.8\pm0.1$ MiB
  & $12$ & $O(n^{1/2})$ & $12.8\pm0.1$ MiB \\
  & $22$ & $O(n^{1/2})$ & $12.8\pm0.1$ MiB
  & $13$ & $O(n^{1/2})$ & $12.8\pm0.1$ MiB \\
  & $24$ & $O(n^{1/2})$ & $12.8\pm0.1$ MiB
  & $14$ & $O(n^{1/2})$ & $12.8\pm0.1$ MiB \\
\bottomrule


\toprule
\multicolumn{1}{c|}{\multirow{2}{*}{\textbf{Scheme}}} 
  & \multicolumn{3}{c|}{\bf Vary $t$} 
  & \multicolumn{3}{c}{\bf Vary $\mathbb{F}_p$} \\
\cmidrule(lr){2-4} \cmidrule(lr){5-7}
  & ~$t$~ & Asymptotic & Experimental
  & ~~$\log|\mathbb{F}_p|$~~ & Asymptotic & Experimental \\
  
\midrule

\multirow{4}{*}{$\Gamma_1$}
  & $1$  & $O(n^{1/12})$ & $1.0\pm0.1$ MiB
  & $16$  & $O(n^{1/12})$ & $152 \pm 10 $ KiB \\
  & $2$  & $O(n^{1/6})$  & $3.2\pm0.1$ MiB
  & $32$  & $O(n^{1/12})$ & $267 \pm 10 $ KiB \\
  & $3$  & $O(n^{1/4})$  & $11.5\pm0.1$ MiB
  & $64$  & $O(n^{1/12})$ & $508 \pm 10 $ KiB \\
  & $4$  & $O(n^{1/3})$  & $41.0\pm0.3$ MiB
  & $128$ & $O(n^{1/12})$ & $1.0\pm0.1$ MiB\\
\midrule
\multirow{4}{*}{$\Gamma_2$}
  & $1$  & $O(n^{1/4})$  & $3.6\pm0.1$ MiB
  & $16$  & $O(n^{1/4})$  & $465 \pm 10$ KiB \\
  & $2$  & $O(n^{1/2})$  & $315.4 \pm 5$ MiB
  & $32$  & $O(n^{1/4})$  & $932 \pm 20$ KiB \\
  & $3$  & $O(n)$     & Too large  
  & $64$  & $O(n^{1/4})$  & $1.8\pm0.1$ MiB \\
  & $4$  & $O(n)$     & Too large  
  & $128$ & $O(n^{1/4})$  & $3.6\pm0.1$ MiB \\
\midrule
\multirow{4}{*}{\cite{goldberg2007improving}}
  & $1$  & $O(n^{1/2})$ & $12.8\pm0.1$ MiB
  & $16$  & $O(n^{1/2})$ & $1.6\pm0.1$ MiB \\
  & $2$  & $O(n^{1/2})$ & $12.8\pm0.1$ MiB
  & $32$  & $O(n^{1/2})$ & $3.2\pm0.1$ MiB \\
  & $3$  & $O(n^{1/2})$ & $12.8\pm0.1$ MiB
  & $64$  & $O(n^{1/2})$ & $6.4\pm0.1$ MiB \\
  & $4$  & $O(n^{1/2})$ & $12.8\pm0.1$ MiB
  & $128$ & $O(n^{1/2})$ & $12.8\pm0.1$ MiB \\
\midrule
\multirow{4}{*}{\cite{devet2012optimally}}
  & $1$  & $O(n^{1/2})$ & $12.8\pm0.1$ MiB
  & $16$  & $O(n^{1/2})$ & $1.6\pm0.1$ MiB \\
  & $2$  & $O(n^{1/2})$ & $12.8\pm0.1$ MiB
  & $32$  & $O(n^{1/2})$ & $3.2\pm0.1$ MiB \\
  & $3$  & $O(n^{1/2})$ & $12.8\pm0.1$ MiB
  & $64$  & $O(n^{1/2})$ & $6.4\pm0.1$ MiB \\
  & $4$  & $O(n^{1/2})$ & $12.8\pm0.1$ MiB
  & $128$ & $O(n^{1/2})$ & $12.8\pm0.1$ MiB \\
\bottomrule
\end{tabular}
\caption{Asymptotic and Experimental Per‐Server Communication Complexities under Variations of $k$, $b$, $t$, and $\log|\mathbb{F}_p|$, Averaged over 100 Trials. The Baseline Configuration is $(k,b,t,\log|\mathbb{F}_p|)=(20,12,1,128)$.}
\label{tab:vary k, b, t, F}
\end{table}

\clearpage

\vspace{5pt}
The experiments are divided into two phases. 
In Phase I, 
we employ a database of size $n=2^{26}$ over $\mathbb{F}_p$, with initial parameters $(k,b,t) = (20,12,1)$ and $p = 2^{128} + 51$. Thereafter, we individually vary $k$, $b$, $t$, and $\log p$, derive the asymptotic communication complexity per server, and measure the experimental communication volume. Throughout, for $\Gamma_1$ we set
$
w_1 = 2(k - b - 2)/t.
$ 
To minimize errors caused by randomness, each experiment is repeated $100$ times, and the average results are recorded as in {\bf Table.} \ref{tab:vary k, b, t, F}.
As evidenced by the data, although schemes $\Gamma_1$ and $\Gamma_2$ fall short of the protocols in \cite{goldberg2007improving} and \cite{devet2012optimally} for larger values of $t$, they incur markedly lower communication overhead when $t$ is small. Moreover, this benefit grows more pronounced as the total number of servers increases and the number of Byzantine servers decreases.
In Phase II, 
we consider the database of size \(n = 2^{16}\) over \(\mathbb{F}_p\), with parameters \((k,b,t) = (6,3,1)\). 
We instantiate two variants over distinct prime fields, namely \(p = 2^7 + 3\) and \(p = 2^{10} + 7\). For each configuration, the scheme is executed \(10^6\) times, and the worst-case cardinality of the output list is reported in {\bf Table.}~\ref{tab:list size}.
It is evident that the maximum list sizes of schemes \(\Gamma_1\) and \(\Gamma_2\) remain small and exhibit negligible variation as the field \(\mathbb{F}_p\) grows.  In contrast, the protocols of \cite{goldberg2007improving} and \cite{devet2012optimally} display a marked increase in list size.  Consequently, one can surmise that for sufficiently large \(\mathbb{F}_p\) there exist parameter choices \((k,b,t)\) and adversarial response strategies under which those earlier schemes yield substantially larger lists, whereas our constructions continue to produce output lists whose size remains bounded independently of \(p\).

\begin{table}[t]
    \centering
    \begin{tabular}{c|cc|cc}
    \toprule
       \multicolumn{1}{c|}{~~~\textbf{Scheme}~~~} 
       & ~~~~$\log |\mathbb{F}_p|$~~~~ &  ~~~\makecell{Worst List Size}~~~ 
       & ~~~~$\log |\mathbb{F}_p|$~~~~ &  ~~~\makecell{Worst List Size}~~~\\ 
    \midrule
       ~~$\Gamma_1$~~                   & 7 &  2 & 10 &   2\\ 
       ~~$\Gamma_2$~~                   & 7 &  3 & 10 &   4\\ 
       ~~\cite{goldberg2007improving}~~ & 7 &  8 & 10 &   15\\ 
       ~~\cite{devet2012optimally}~~    & 7 &  6 & 10 &   13\\
    \bottomrule
    \end{tabular}
    \caption{Worst-case output list sizes (over $10^6$ runs) for four ldBRPIR schemes with parameters $(k,b,t) = (6,3,1)$ and database size $n = 2^{16}$, instantiated over two prime fields of size $p = 2^7 + 3$ and $2^{10}+7$ (i.e. $\log |\mathbb{F}_p| = 7$ and $10$).}
    \label{tab:list size}
\end{table}

\section{Conclusion}
    We have introduced two perfect list-decodable BRPIR schemes, \( \Gamma_1 \) and \( \Gamma_2 \). The Byzantine tolerance bound of scheme \( \Gamma_1 \) significantly surpasses those established in \cite{goldberg2007improving} and \cite{devet2012optimally}.
    Compared to the schemes proposed in \cite{goldberg2007improving} and \cite{devet2012optimally}, our schemes demonstrate a notable improvement in communication complexity. This advantage is particularly pronounced in scenarios with a small and fixed privacy threshold, such as \( t = O(1) \). We achieve a small maximum output list size that is independent of the size of the finite field containing the data and depends solely on the number of servers. Additionally, our scheme \( \Gamma_1 \) offers higher Byzantine robustness than \cite{goldberg2007improving} and \cite{devet2012optimally}.

\clearpage

%
%
%
%

\bibliographystyle{splncs04}
\bibliography{bib}





\end{document}